\documentclass[12pt]{article}
\usepackage{chatsagnik}
\usepackage[acronym]{glossaries-extra}
\setabbreviationstyle[acronym]{long-short}

\newcommand{\ovl}[1]{\overline{#1}}
\newcommand{\ftwon}[1]{\{0,1\}^{#1}}
\newcommand{\ippwc}{\mathcal{W}}
\newcommand{\lig}{L_{i,g}}

\newcommand{\xcancel}[1]{}
\newcommand{\st}[1]{}

\newcommand{\pwc}{\textsc{Pw}}

\newcommand{\abpo}{O_{\text{h}}}
\newcommand{\iglalgo}{{\tt IGL}}

\newcommand{\xor}{\oplus}

\newcommand{\qgl}{\textsc{QGL}}
\newcommand{\mae}{\textsc{MAE}}
\newacronym{mq}{MQ}{Membership Query}
\newacronym{qex}{\textsc{Qex}}{Quantum Example}
\newacronym{qaex}{\textsc{Qaex}}{Quantum Agnostic Example}
\newcommand{\mq}{{\gls*{mq}}~}
\newcommand{\qex}{{\gls*{qex}}~}
\newcommand{\qaex}{{\gls*{qaex}}~}

\title{Efficient Quantum Agnostic Improper Learning of Decision Trees}
\author[1]{Sagnik Chatterjee}
\author[1]{Tharrmashastha SAPV}
\author[1,2]{Debajyoti Bera}
\affil[1]{Indraprashta Institute of Information Technology (IIIT-D), Delhi, India\thanks{\textit {\{sagnikc,tharrmashasthav,dbera\}@iiitd.ac.in}}}
\affil[2]{Centre for Quantum Technologies, IIIT-Delhi}
\date{}
\begin{document}
\maketitle
\begin{abstract}
  The agnostic setting is the hardest generalization of the PAC model since it is akin to learning with adversarial noise. 
  In this paper, we give a $\mathrm{poly}\left(n,{t},{\nicefrac{1}{\varepsilon}}\right)$ quantum algorithm for learning size $t$ decision trees over $n$-bit inputs with uniform marginal over instances, in the agnostic setting, \textit{without membership queries} (MQ). 
  This is the first algorithm (classical or quantum) for efficiently learning decision trees without MQ. First, we construct a quantum agnostic weak learner by designing a quantum variant of the classical Goldreich-Levin algorithm that {works with strongly biased function oracles}. 
  Next, we show how to quantize the agnostic boosting algorithm by Kalai and Kanade (2009) to obtain the \textit{first} efficient quantum agnostic boosting algorithm (that has a \textit{polynomial speedup} over existing adaptive quantum boosting algorithms). 
  We then use the quantum agnostic boosting algorithm to boost the weak quantum agnostic learner constructed previously to obtain a quantum agnostic learner for decision trees. Using the above framework, we also give quantum decision tree learning algorithms without MQ in weaker noise models.
\end{abstract}
\clearpage
\tableofcontents
\section{Introduction}
Efficiently learning decision trees is a central problem in algorithmic learning theory since any Boolean function is learnable as a decision tree ~\citep{bshouty93monotone}. There has been a large body of work (see \cref{table:querycomplexity}) centered around providing theoretical guarantees for learning decision trees under various generalizations and restrictions of the \pac model introduced by~\citet{valiant1984theory}. 

The original \pac model~\citep{valiant1984theory} is in the noiseless setting where the learning algorithm is trained on a training set $S=\{(x_i,y_i)\}_{i\in[m]}$ consisting of $m$ tuples of instances $x_i\in\mathbb{F}^n_2$ and their corresponding binary labels $y_i$. In the random classification noise (RCN) setting, the learning algorithm is trained on a set $S^{\prime}$ where each label $y_i$ in $S$ is flipped with a uniform probability $p$. In the agnostic setting (adversarial noise), each label in $S$ is flipped with some probability which is dependent on the example. 

There are two types of decision tree learning algorithms: \textit{proper learning} algorithms, where the output is a decision tree, and \textit{improper learning} algorithms, where the output hypothesis is not necessarily required to be a decision tree. Proper learning of decision trees, \textit{even in the noiseless setting}, is known to be computationally hard~\citep{koch2023superpolynomial}, and all the efficient improper learning algorithms (for different noise models) are designed to use \mq oracles (see \cref{table:querycomplexity}).

\noindent\textbf{Downsides of MQ oracles.} A \mq oracle allows a learning algorithm to fetch the label of {\em any desired instance} in the input space, even among the ones absent in the training set. In the famous experiment by \citet{baum1992query}, the \mq oracle was queried by the learning algorithm on instances outside the domain of the labeling function. This makes MQ oracles difficult to implement and is probably one reason that makes them unattractive to the applied machine learning community~\citep{bshouty2002using,awasthi2013learning}, which brings us to the main question tackled in this work.\vspace{-0.1cm}
    \begin{center}
    \doublebox{%
    \begin{minipage}{22em}
    \textbf{Question:} Does there exist a polynomial time (improper) decision tree learning algorithm without membership queries?
    \end{minipage}}
    \end{center}

\noindent\textbf{Quantum as the silver bullet.} In practice, machine learning algorithms use data in the training set to learn a hypothesis. This setup can be modeled as having query access to a random example oracle where we sample training points according to the uniform distribution. Theoretically, it is known that the PAC+MQ model is strictly stronger than the \pac model with only random examples~\citep{angluin1988queries,bshouty93monotone,feldman2006optimal,valiant1984theory}. Similar to the random example oracle, access to a uniform superposition over the training set is an equivalent and a natural requirement in quantum computing. This was first demonstrated by ~\citet{bshouty1998learning} where they introduced the notion of the Quantum \pac model. Many subsequent works (see~\citet{atici2007quantum,Arunachalam2020,deWolf2020,chatterjee2023quantum}) have been designed in the realizable \textit{quantum} \pac model with access to a uniform superposition over the training examples. 
It is not known whether random examples are sufficient for any (quantum or classical) agnostic learning task, which was another motivation behind this work.

The query models used in our quantum algorithm for improperly learning decision trees were proposed by \citet{bshouty1998learning} and \citet{arunachalam2017guest}; and are generalizations of the random example oracle where the learning algorithm has query access to a superposition over all instances in the domain. A detailed description of the \qex and \qaex oracles is given in \cref{sec:prelims}. While the random example query model is weaker than the \qex model, the \mq model is stronger than the \qex model w.r.t. uniform marginal distribution~\citep{bshouty1998learning}. 
\begin{table*}[t]
  \footnotesize
  \centering
  \caption{\footnotesize 
Comparing different algorithms for learning size-$t$ decision trees on $n$-bit Boolean functions. Note here that $t$ and $\nicefrac{1}{\varepsilon}$ can be as large as $\mathrm{poly}(n)$, which renders the running time of many of the algorithms given below as {  super-polynomial}. Our quantum algorithms are {  strictly polynomial in all parameters} while being the only algorithm to {  work in the agnostic and realizable settings} and {  not use membership queries} (denoted by MQ). Here we note that the number of training samples $m$ required for learning is $\mathrm{poly}(n)$. QC denotes query complexity.
}
  {
    \renewcommand{\arraystretch}{2} \footnotesize
    \begin{tabular}
      {p{3cm}@{} p{2cm} p{2cm} p{2cm} p{1cm} p{2cm} p{2.1cm}} 
      
       \textbf{Work} & \textbf{Setting}&\textbf{Type} &\textbf{Noise Setting}&\textbf{MQ}&\textbf{Runtime}&\\
       \toprule
       EH~\citeyear{ehrenfeucht1989learning} &Classical & Proper & \begin{minipage}{.5\textwidth}{  Realizable}\end{minipage} & \begin{minipage}{.5\textwidth}{   No}\end{minipage} & \begin{minipage}{.5\textwidth}
    {  ${\mathrm{poly}\left(n^{\log t},{\nicefrac{1}{\varepsilon}}\right)}$}
\end{minipage}&\\

       KM~\citeyear{kushilevitz1991learning}&Classical & Improper & \begin{minipage}{.5\textwidth}{  Realizable}\end{minipage} & \begin{minipage}{.5\textwidth}{  Yes}\end{minipage} & \begin{minipage}{.5\textwidth}
     {  ${\mathrm{poly}\left(n,t,{\nicefrac{1}{\varepsilon}}\right)}$}
 \end{minipage}&\\

         LMN~\citeyear{LMN93}&Classical & Proper & \begin{minipage}{.5\textwidth}{  Realizable}\end{minipage} & \begin{minipage}{.5\textwidth}{   No}\end{minipage}& \begin{minipage}{.5\textwidth}{ ${\mathrm{poly}\left(n^{\log{\left({t}/{\varepsilon}\right)}}\right)}$}
 \end{minipage}&\\

     MR~\citeyear{MR02}&Classical  & Proper & \begin{minipage}{.5\textwidth}{   Agnostic}\end{minipage} & \begin{minipage}{.5\textwidth}{   No}\end{minipage} & \begin{minipage}{.5\textwidth}{ ${\mathrm{poly}\left(n^{\log{\left({t}/{\varepsilon}\right)}}\right)}$}
 \end{minipage}&\\
 \cmidrule(l{1.75em}r{1.5em}){1-6}
        GKK~\citeyear{gopalan2008agnostically}
        & \multirow{3}{*}{Classical}  
        & \multirow{3}{*}{Improper} 
        &\multirow{3}{*}{Agnostic} 
        &\multirow{3}{*}{Yes} 
        & \multirow{3}{*}{${\mathrm{poly}\left(n,t,{\nicefrac{1}{\varepsilon}}\right)}$}
        &\\
        KK~\citeyear{kalai-kanade}&&&&&&\\ Feldman~\citeyear{feldman2009}&&&&&&\\
        \cmidrule(l{1.75em}r{1.5em}){1-6}
        BLT~\citeyear{BLT20}&Classical & Proper & \begin{minipage}{.5\textwidth}{   Agnostic}\end{minipage} & \begin{minipage}{.5\textwidth}{   No}\end{minipage} &\begin{minipage}{.5\textwidth}{  ${\mathrm{poly}\left(n^{\log t},{\nicefrac{1}{\varepsilon}}\right)}$}\end{minipage}&\\
      \midrule
      \multirow{2}{\columnwidth}{\textbf{This Work}} &Quantum & Improper & \begin{minipage}{.5\textwidth}{   Realizable}\end{minipage} & \begin{minipage}{.5\textwidth}{   No}\end{minipage} & \begin{minipage}{.5\textwidth}{   ${\mathrm{poly}\left(n,{t},{\nicefrac{1}{\varepsilon}}\right)}$}\end{minipage}&\textbf{QC:} $\bigO{{1}/{\varepsilon^2}}$\\
      \cmidrule(r){2-7}
      &Quantum & Improper & \begin{minipage}{.5\textwidth}{   Agnostic}\end{minipage} & \begin{minipage}{.5\textwidth}{   No}\end{minipage} & \begin{minipage}{.5\textwidth}{   ${\mathrm{poly}\left(n,{t},{\nicefrac{1}{\varepsilon}}\right)}$}\end{minipage}&\textbf{QC:} $\bigO{{n^2}/{\varepsilon^3}}$\\
      
      \bottomrule
    \end{tabular}
  }
  \label{table:querycomplexity}
\end{table*}
\subsection{Our Contributions and Technical Overview}\label{sec:overview}

The main contribution of this work is a quantum polynomial time algorithm for improperly learning decision trees {\em without} \mq in the agnostic setting (and hence, in weaker noise settings). The importance is twofold.
\begin{enumerate}
    \item To our knowledge, ours is the first {\em quantum algorithm} for decision tree learning (realizable or agnostic, with or without MQ).
    \item Our algorithm is also the only known {\em efficient} agnostic PAC learning algorithm for decision trees (classical or quantum) {\em without MQ}~\footnote{Our result subsumes the {classical} realizable learning algorithm for \textit{monotone} decision trees without MQ by~\citet{o2007learning}.}.
\end{enumerate}
We state a simplified version of our main result now.
\begin{theorem}\label{thm:maininfo}
Given $m$ training examples, there exists a quantum algorithm for learning size-$t$ decision trees in the agnostic setting without \mq in $\mathrm{poly}\left(m,{t},\nicefrac{1}{\varepsilon}\right)$ time.
\end{theorem}
Here we note that the number of training samples $m$ required for learning is polynomial w.r.t. to $n$ where the decision trees correspond to $n$-bit Boolean functions. Following earlier work (see \cref{sec:related}), we also {assume} a uniform marginal distribution over the instances.  
In \cref{table:querycomplexity}, we compare our decision tree learning algorithm against existing decision tree learning algorithms. Our algorithm (see \cref{fig:algo-overview}) follows from the existence of
\begin{itemize}
    \item an efficient quantum agnostic boosting algorithm (see \cref{sec:qagnosticalgo}), and 
    \item an efficient weak quantum agnostic learner for decision trees (see \cref{sec:QDTL}).
\end{itemize}
\begin{figure}[t]
    \centering
    \includegraphics[width=0.7\textwidth]{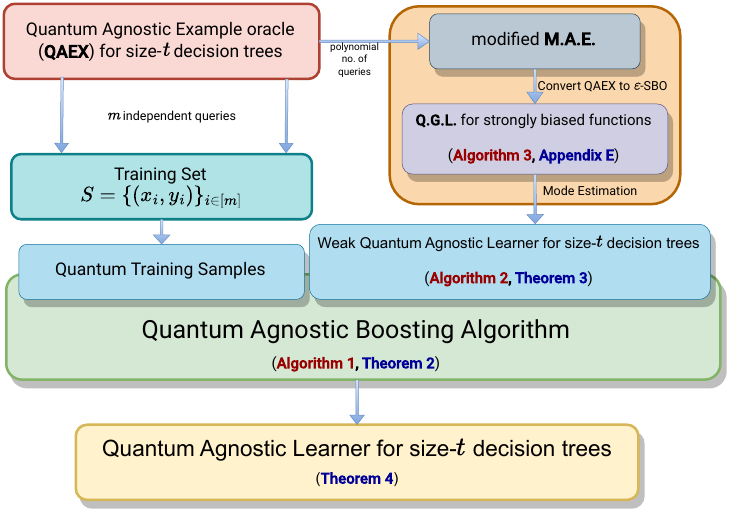}
    \small
    \caption{\footnotesize Agnostically learning polynomial-sized decision trees without MQ.}
    \label{fig:algo-overview}
\end{figure}
\subsubsection{Quantum agnostic boosting}\label{sec:qagboostex} Quantum boosting algorithms for the realizable setting have been shown to exist (see \cref{sec:related}), but their existence in the agnostic setting was an open question~\citep{deWolf2020}. The challenge in such algorithms is precisely estimating the margins under the presence of instance-dependent noise. Our idea was to quantize the \citet{kalai-kanade} (KK) algorithm whose use of \textit{relabeling} let us avoid using the amplitude amplification subroutine (a staple in the previous quantum boosting algorithms) explicitly, thereby removing a significant source of error. In \cref{sec:qagnosticalgo}, we show that given a weak quantum agnostic learner $A$ with an associated hypothesis class $\mathcal{C}$ and a set of $m$ training examples, we can construct a  $\mathrm{poly}\left(m,\nicefrac{1}{\varepsilon}\right)$ time quantum boosting algorithm to produce a hypothesis that is $\varepsilon$ close to the best hypothesis in $\mathcal{C}$.

\subsubsection{Weak Quantum Agnostic Learner for Decision Trees}\label{sec:qweakex} 
In \cref{sec:QDTL}, we construct a quantum weak agnostic learner for size-$t$ decision trees using $\bigO{n^2/\varepsilon^3}$ queries to the \qaex oracle instead of the \mq oracle. The weak learner is constructed using a new quantum variant of the Goldreich-Levin algorithm (GL)~\citep{goldreich1989hard}. We use QGL to identify the monomial that best approximates the Bayes optimal predictor. This monomial serves as our weak learner. We are aware of only one prior quantum Goldreich-Levin algorithm~\citep{adcock2002quantum}; however, that algorithm involved different types of oracles and tackled a problem unrelated to ours. We now briefly touch upon the technical challenges encountered.

\begin{enumerate}
    \item The classical GL algorithm requires obtaining $f(x)$ for specific instances $x$. Our \qgl{} algorithm (see \cref{alg:qgl}), instead, was designed to work with the \qaex oracle, which generates a superposition over all $(x,f(x))$ pairs. The key step in the \qgl{} algorithm is using a Deutsch-Jozsa-style sampler to work in tandem with the \qaex oracle. This brings us to the second technical challenge.
    \item The true label $f(x)$ of any $x$ is imperative for the classical Goldreich-Levin algorithm to work properly. However, in the agnostic scenario, both $f(x)$ (correct label) and $1-f(x)$ (incorrect label) may be returned with non-zero probability. The probabilities could also depend on $x$~, which makes matters worse. Thus, we designed a wrapper around \qaex denoted $O_h$ (see ~\cref{alg:QWeakLearner}) employing the recent technique of multi-distribution amplitude estimation (\mae)~\citep{bera2022few} to ensure a bound on the errors. 
    
    We note here that oracles in which the probability of label flips do not depend on $x$ capture the RCN model and have been studied as \textit{biased oracles}. To differentiate, we refer to oracles where the probability of label flips are dependent on the instance as \textit{strongly biased oracles} --- these capture the agnostic setting.
    \item Our QGL algorithm is run on the wrapper oracle $O_h$. Unfortunately, the QGL algorithm itself uses erroneous subroutines like amplitude amplification and estimation. Such algorithms often exhibit grossly incorrect behaviors (e.g., the amplification step may amplify the amplitudes of even the undesired states due to the error arising from amplitude estimation). We meticulously ensured amplitude amplification and estimation work in tandem to keep their inherent errors in control, particularly as the algorithm proceeds to lower levels of the prefix search tree (where the errors have a chance to accumulate). 
\end{enumerate}

\noindent{\bf Weaker Noise Settings.} The agnostic setting generalizes the realizable and the random classification noise (RCN) settings; thus, our framework also learns decision trees in those settings, as explained in \cref{sec:QTDLreal} and \cref{sec:QTDLrandom} respectively.

\subsection{Related Work}\label{sec:related}
\textbf{Agnostic Boosting}. 
\citet{kalaiParity2008} gave a classical agnostic boosting algorithm that achieves nearly optimal accuracy. We follow the agnostic boosting formalization of \citet{kalaiParity2008}(as opposed to earlier works like \citet{ben-david2001,gavinsky2002}) in this paper. \citet{feldman2009}, and KK~\citeyearpar{kalai-kanade} came up with distribution-specific agnostic boosting algorithms to circumvent certain impossibility results on convex boosting algorithms~\citep{long-servedio}. We give a quantum version of the KK algorithm that also achieves a quadratic speedup in the VC dimension of the weak learner.
\\\vspace{-0.2cm}

\noindent\textbf{Agnostic Learning of Decision Trees}. 
~\citet{ehrenfeucht1989learning} gave the first \textit{weakly} proper learning algorithm with quasi-polynomial running time and sample-complexity in the realizable setting using random examples. Subsequent works on properly learning decision trees\citep{MR02,BLT20,guy20universalcriterion,guy22proplearningjacm} either have quasi-polynomial dependence on error parameters and intensive memory requirements or require the use of \mq (see \cref{table:querycomplexity}). Recently, it was shown by \citet{koch2023superpolynomial} that efficient proper learning of decision trees has a superpolynomial lower bound. \citet{bshouty2023superpolynomial} showed that the superpolynomial lower bound also holds for proper learning of monotone decision trees.

\citet{kushilevitz1991learning} gave the first polynomial time {improper} decision tree learning algorithm (we henceforth refer to this as the KM algorithm) using \mq in the realizable setting. Their approach was later extended to the agnostic setting by \citet{gopalan2008agnostically,kalai-kanade,feldman2009}. To our knowledge, there is no prior work on quantum agnostic learning.

\noindent\textbf{Quantum Boosting}. \citet{Arunachalam2020} gave the first quantum adaptive boosting algorithm, which was a quantum generalization of the celebrated AdaBoost algorithm. Their approach was later extended to work on non-binary weak learners by \citet{chatterjee2023quantum}. 
 Both of the above boosting algorithms generate a quadratic speedup compared to their classical counterparts in the VC dimension of the weak learner. This speedup is retained by our quantum agnostic boosting algorithm.
\section{Notation and Preliminaries}
\label{sec:prelims}

\par\noindent\textbf{Fourier Analysis of Boolean Functions}.
Given any Boolean function $f:\mathbb{F}_2^n\xrightarrow{}\{-1,1\}$, , where $\mathbb{F}^n_2=\{0,1\}^n$, we can uniquely express it as $f(x)=\sum_{S\in\mathbb{F}^n_2}\hat{f}(S)\chi_S(x)$. Here $\hat{f}(S)=\expect{f(x)\chi_S(x)}=\langle f,\chi_S\rangle$ are the Fourier coefficients corresponding to every $S$, and $\chi_S(x)=\prod_{i\in S}(-1)^{x_i}$, where $x_i$ are 0-1 valued. $\chi_S(x)$ is the multilinear monomial corresponding to every $S$ (also referred to as the parity of $S$). For Boolean functions, the squares of the Fourier coefficients $\hat{f}^2(S)$ form a probability distribution.

In algorithmic learning, our objective is to learn an approximation of the Fourier representation\footnote{For a detailed survey on the connection between Fourier representation and learning theory see~\cite{mansour1994learning}.} of $f$ by finding the set of strings $S$ that have high $\hat{f}(S)$ values.
We design a quantum variant (see \cref{alg:qgl}) of the classical GL algorithm~\citep{goldreich1989hard} to find terms with Fourier coefficients larger than a threshold $\tau$. The QGL algorithm searches a binary tree 
of all possible prefixes of $n$-length strings; the root corresponds to the empty prefix, and the leaves correspond to complete strings, s.t. every string represents a monomial. The weight of a node $a$ of length $s$ is defined as $\pwc (a)=\sum_{b\in \{0,1\}^{n-s}} \hat{f}^2(ab)$.\\
\par\noindent\textbf{Agnostic \pac Learning.}
Consider an $n$-bit function or ``concept" $c\in\mathcal{C}:\mathbb{F}^n_2\xrightarrow{} \{-1,1\}$. 
In the agnostic setting~\citep{haussler1992decision,kearns-agnostic}, a learning algorithm tries to learn some unknown concept w.r.t. a fixed arbitrary joint distribution $\mathcal{D}$ over $\mathbb{F}^n_2\times\{-1,1\}$. The agnostic setting is seen as learning with adversarial noise in the following manner: Let $\mathcal{D}$ be a joint distribution over the examples and the labels $\mathbb{F}^n_2\times\{-1,1\}$. We can also interpret this as a distribution $\mathcal{D}^{\prime}$ over $\mathbb{F}^n_2$, where the examples are labeled according to some concept $c^{\prime}$, s.t. an adversary corrupts some $\eta$ fraction of the labels given to the algorithm. In the agnostic setting, training error of a hypothesis $h$, $\mathrm{err}_S(h)= {\mathrm{Pr}_{S}}\left[h(x)\neq y\right]$ is defined w.r.t. set $S$ of $m$ labeled training examples sampled from a joint distribution $\mathcal{D}$ over $\mathbb{F}^n_2\times\{-1,1\}$.
The generalization error is defined as $\mathrm{err}_\mathcal{D}(h)= {\mathrm{Pr}_\mathcal{D}}\left[h(x)\neq y\right]$. 
Correlation is defined as follows.
\begin{defn}[Correlation~\citep{kalai-kanade}]\label{def:correlation}
The correlation of a hypothesis $h\in\mathcal{H}$ w.r.t. $\mathcal{D}$ over $\mathbb{F}^n_2\times\{-1,1\}$ is defined as $\cor{h}{\mathcal{D}}=1-2\mathrm{err}_\mathcal{D}(h)={\mathbb{E}}_{\mathcal{D}}\left[h(x)\cdot y\right]$.
\end{defn}
\vspace{-0.2cm}\noindent The \textit{optimal correlation} of a class of concepts $\mathcal{C}$ is defined as $\mathrm{optcor}_\mathcal{D}(\mathcal{C})=\cor{h^*}{\mathcal{D}}={\mathrm{argmax}_{h\in \mathcal{C}}}\;\;\cor{h}{\mathcal{D}}$.

In agnostic \pac learning, we fix some concept class $\mathcal{C}$ (e.g., decision trees of fixed depth) and aim to learn a hypothesis $h$ close to the best possible concept $c_{\mathrm{opt}}\in\mathcal{C}$. Note that $h$ may not belong to $\mathcal{C}$, as in improper learning. Boosting algorithms are an important class of improper learning algorithms. 
\\\vspace{-0.2cm}
\par\noindent\textbf{Agnostic Boosting}. As discussed earlier, computational hardness results for polytime proper learning led researchers to try the improper learning approach via boosting, where they would take a ``weak"-agnostic learner and boost it to obtain a better (not necessarily optimal as in the realizable case) generalization performance. 
We make these notions precise below.
\begin{defn}[$(m,\kappa,\eta)$-weak Agnostic Learner~\citep{kalai-kanade}]
For some $\kappa=\bigO{1/poly(m)}$, an algorithm $A$ learns concept class $\mathcal{C}$ over an arbitrary distribution $\mathcal{D}$ on $\mathcal{X}\times\{-1,1\}$, on $m$ examples drawn i.i.d. from $\mathcal{D}$, and outputs a hypothesis $h$ s.t.
$
\cor{h}{\mathcal{D}}\geq\eta\cdot\mathrm{optcor}_{\mathcal{D}}(\mathcal{C})-\kappa
$.
\end{defn}
\begin{defn}[$\beta$-optimal $(\varepsilon,\delta)$-agnostic PAC learner~\citep{gavinsky2002}]\label{def:betalearn}
    A learning algorithm $A$ $\beta$-optimally learns a concept class $\mathcal{C}$ if for every $\varepsilon,\delta>0$, $0<\beta\leq\nicefrac{1}{2}$, any arbitrary distribution $\mathcal{D}$ over $\mathcal{X}\times\{-1,1\}$, $A$ takes examples drawn i.i.d. from $\mathcal{D}$, and outputs a hypothesis $h$ {\it s.t.} $\cor{h}{\mathcal{D}}\geq \mathrm{optcor}_{\mathcal{D}}(\mathcal{C})-\beta-\varepsilon$ with probability at least $1-\delta$. 
\end{defn}
For brevity, we shall be referring to $\beta$-optimal $(\varepsilon,\delta)$-agnostic PAC learners as $\beta$-optimal agnostic PAC learners. The goal of Agnostic Boosting~\citep{ben-david2001,gavinsky2002} is to produce a $\beta$-optimal learner given a $(m,\kappa,\eta)$-weak agnostic learner.

\citet{kalai-kanade} introduced the concept of training intermediate weak hypotheses on \textit{randomly relabeled examples} (instead of the traditional reweighting schemes based on AdaBoost) to obtain a $\beta$-optimal agnostic learner. In the fully supervised setting (i.e., w.r.t. this paper), the semantic differences between reweighting and relabeling are negligible.
\\
\par\noindent\textbf{Quantum Agnostic Learning}. Classical learners have access to a random example oracle $\mathrm{EX}(f,\mathcal{D})$ for a function $f$ w.r.t. distribution $\mathcal{D}$ over $\mathbb{F}^n_2$, which samples an instance $x$ according to $\mathcal{D}$, and returns a labeled example $\left\langle x,f(x)\right\rangle$. In the agnostic case, learners have access to the oracle $\mathrm{AEX}(\mathcal{D})$ where $\mathcal{D}$ is a joint distribution over instances and labels. An invocation of $\mathrm{AEX}$ returns a labeled instance $\left\langle x,y\right\rangle$ w.r.t. $\mathcal{D}$. In the Quantum PAC model~\citep{bshouty1998learning}, the quantum learners have access to a \textit{quantum example oracle} $\mathrm{QEX}(f,\mathcal{D})$, s.t. each invocation to $\mathrm{QEX}(f,\mathcal{D})$ produces the quantum state $\left( \sum_{x \in \mathbb{F}^n_2}\sqrt{\mathcal{D}(x)} \ket{x, f(x)} \right)$. In the quantum agnostic setting~\citep{arunachalam2018optimal}, quantum learners can access the oracle $\mathrm{\qaex}(\mathcal{D})$, s.t. each invocation of the oracle produces the quantum state $\left( \sum_{(x,y) \in \mathbb{F}^n_2\times{\{-1,1\}}}\sqrt{\mathcal{D}(x,y)} \ket{x, y} \right)$. We now define a $(m,\kappa,\eta)$-weak quantum agnostic learner.
\begin{defn}[$(m,\kappa,\eta)$-Weak Quantum Agnostic Learner]\label{def:corrAgnostic}
    For some $\kappa=\bigO{1/\mathrm{poly}(m)}$, a quantum algorithm $A$ that learns a concept class $\mathcal{C}$ over an arbitrary distribution $\mathcal{D}$ on $\mathcal{X}\times\{-1,1\}$, with at most $m$ calls to a $\qaex(\mathcal{D})$ oracle, and outputs a hypothesis $h$ s.t.
$
\cor{h}{\mathcal{D}}\geq\eta\cdot\mathrm{optcor}_{\mathcal{D}}(\mathcal{C})-\kappa
$.
\end{defn}
 We can similarly define a quantum version of a {$\beta$-optimal agnostic learner}.
\\
\par\noindent\textbf{Useful Quantum Algorithms}.
\begin{lem}[Amplitude Amplification~\citep{bhmt}]
\label{lem:AmplitudeAmplification}
 Let there be a unitary $U$ such that $U\ket{0}=\sqrt{a}\ket{\phi_0}\ket{0}+\sqrt{1-a}\ket{\phi_1}\ket{1}$ for an unknown $a$ such that $a \ge p>0$ for a known $p$. Then there exists a quantum amplitude amplification algorithm that makes $\Theta(\sqrt{p^{\prime}/p})$ expected number of calls to $U$ and $U^{-1}$ and outputs the state $\ket{\phi_0}$ with a probability $p^{\prime}>0$.
\end{lem}
\cref{lem:AmplitudeAmplification} allows us to boost the probability of success of a marked state with a quadratic speedup compared to probabilistic amplification algorithms.
\begin{lem}[Relative Error Estimation~\citep{bhmt}]
Given an error parameter $\varepsilon$, a constant $k\geq1$, and a unitary $U$ such that $U\ket{0}=\sqrt{a}\ket{\phi_0}\ket{0}+\sqrt{1-a}\ket{\phi_1}\ket{1}$ where either ${a}\geq p$ or $a=0$. Then there exists a quantum algorithm that produces an estimate $\Tilde{a}$ of the success probability $a$ with probability at least $1-\frac{1}{2^k}$ such that $\left|a-\Tilde{a}\right|\leq \varepsilon a$ when $a\geq p$. The expected number of calls to $U$ and $U^{-1}$ made by our quantum amplitude estimation algorithm is 
$
  O\left(\frac{k}{\varepsilon\sqrt{p}}\left(1+\log{\log{\frac{1}{p}}}\right)\right) 
$.\label{lem:RelativeEstimation}
\end{lem}
We see that \cref{lem:RelativeEstimation} can be used for mean estimation with a relative error by setting $p=\bigO{1/m}$, where $\ket{\phi_0}$ is a superposition over $m$ basis states. 
This lemma follows from the amplitude estimation lemma (Theorem 15 of~\cite{bhmt}) by setting $t=\frac{1}{m}$.

\begin{lem}[Multidistribution Amplitude Estimation. Theorem 4 of \cite{bera2022few}]
    Given an oracle $O$ that acts as $O\ket{0} = \sum_y\alpha_y\ket{y}(\eta_{0,y}\ket{0}+\eta_{1,y}\ket{1})$, there exists an algorithm to output the quantum state $\sum_y\alpha_y\ket{y}(\eta_{0,y}\ket{0}+\eta_{1,y}\ket{1})\ket{\widetilde{\eta_{1,y}}}$ in $\bigO{1/\varepsilon}$ queries with a high probability, such that $\abs{\widetilde{\eta_{1,y}}-\eta_{1,y}}\leq\varepsilon$.\label{lem:MultiEstimation}
\end{lem}

Given a joint distribution $D$ over $\mathcal{X}\times\mathcal{Y}$,  \cref{lem:MultiEstimation} allows us to estimate the conditional probability 
$\mathrm{Pr}\left[{\mathcal{Y}= y|x}\right]$ to within $\varepsilon$ accuracy over all $x$ in superposition.
 
\section{Quantum Agnostic Boosting}\label{sec:qagnosticalgo}
\begin{algorithm}[t!]
\setcounter{AlgoLine}{0}
    \caption{Quantum Agnostic Boosting}
    \label{alg:QAgBoost}
    \DontPrintSemicolon
    \SetKwComment{Comment}{~$\vartriangleright$~}{}
    \LinesNumbered
    \footnotesize
    \SetKwInput{kwInit}{Initialize}
    \KwIn{$(m,\kappa,\eta)$-weak quantum agnostic learner $A$ and its corresponding \qaex oracle.}
    \kwInit{$H^0=0$, $\varepsilon>0$, $T=\bigO{\nicefrac{1}{\eta^2\varepsilon^2}}$. Prepare a set $S$ of $m$ training samples $\{(x_i,y_i)\}_{i\in[m]}$ by measuring the output of \qaex.}
    \KwOut{Hypothesis $H^{\hat{t}}$ for some $\hat{t} \in \{1,2, \ldots T\}$ such that $\textrm{err}_S(H^{\hat{t}})=\min_t \textrm{err}_S(H^{{t}})$.}
    \For{$t=1$ to $T$}{
            {Prepare $2+m$ copies of
            of the state 
            $\ket{\psi_0}=\nicefrac{1}{\sqrt{m}} \sum_{i \in [m]} \ket{x_i,y_i}$.}\;
            
            {Query the oracle $O_{H^{t-1}}$ to obtain $2+m$ copies of the state       $\nicefrac{1}{\sqrt{m}}\sum_{i\in[m]}\ket{x_i,y_i}{\ket{w^t_i}}$
            .}\;
            
            {On the last $m$ copies, perform arithmetic operations to obtain $\nicefrac{1}{\sqrt{m}}\sum_{i\in[m]}\ket{x_i,y_i}\ket{z_i}$.}\\\Comment{Let $z_i=\nicefrac{\left(1+w^t_i\right)}{2}$,  $z^\prime_i=\nicefrac{\left(1-w^t_i\right)}{2}$.}
            
            {Obtain $\ket{\phi_3}$ by a conditional rotation on $\ket{z_i}$. $\ket{\phi_3}=\frac{1}{\sqrt{m}}\sum_{i\in[m]}\ket{x_i,y_i}\left(\sqrt{z_i}\ket{0}+\sqrt{z^\prime_i}\ket{1}\right)$}{\label{line:phi3}.}\;
            
            {Perform a \textsc{Cnot} operation on $\ket{y_i}$ with the last register as control to obtain $m$ copies of $\ket{\phi_4}=\frac{1}{\sqrt{m}}\sum_{i\in[m]}\ket{x_i}\left(\sqrt{z_i}\ket{y_i,0}+\sqrt{z^\prime_i}\ket{\bar{y}_i,1}\right)$.}{\label{line:phi4}}\;
            \Comment{Denote the unitary for obtaining $\ket{\phi_4}$ by $\textsc{Qaex}_{t}$.}            
            {Obtain oracle $O_{h^t}$ corresponding to hypothesis $h^t$ produced by weak learner $A$ using $\textsc{Qaex}_t$ as the quantum example oracle instead of \qaex.}
            
            {Invoke $O_{h^t}$ on the 1st copy of $\ket{\phi_0}$ to obtain $\frac{1}{\sqrt{m}}\sum_{i\in[m]}\ket{x_i,y_i,w^t_i,h^t(x_i)}$.}\\ \Comment{Let $\alpha_t$ $=$ $\frac{1}{{m}}\sum_{i\in[m]}\left(w^t_i y_i h^t(x_i)\right)$.}
            { Prepare the state $$\sqrt{1-\alpha_t}\ket{\psi_0,0}+\sqrt{\alpha_t}\ket{\psi_1,1}.$$ Estimate $\alpha_t$ as $\tilde{\alpha}_t$.}{\label{line:alpha}}\;
            
            {Invoke $O_{h^t}$ on the 2nd copy of $\ket{\phi_0}$ to obtain $\frac{1}{\sqrt{m}}\sum_{i}\ket{x_i,y_i,w^t_i}\ket{-H^{t-1}(x_i)}$.}\\ \Comment{Let $\beta_t=\nicefrac{1}{{m}}\sum_{i\in[m]}\left(w^t_i\cdot y_i\cdot -H^{t-1}(x_i)\right)$.} 
            {Prepare the state $$\sqrt{1-\beta_t}\ket{\psi_0,0}+\sqrt{\beta_t}\ket{\psi_1,1}.$$ Estimate $\beta_t$ as $\tilde{\beta}_t$.}{\label{line:beta}}\;
            
            {If $\tilde{\alpha}_t>\tilde{\beta}_t$, $H^t=H^{t-1}+\tilde{\alpha}_t\cdot h^t$. Otherwise,  $H^t=\left(1-\tilde{\beta}_t\right) H^{t-1}$. Construct the oracle $O_{H^t}$.}\;\label{line:1_algo_boosting}
        }
        {Return the ${H^{\hat{t}}}$ with the least training error on $S$ for $\hat{t} \in [T]$.}
\end{algorithm}

In this section, we describe our quantum agnostic boosting algorithm that has query access to a $(m,\kappa,\eta)$-weak quantum agnostic learner $A$, and to its corresponding $\qaex(D)$ oracle for an unknown joint distribution $D$ over $\mathcal{X}\times\{-1,1\}$. As is common in quantum boosting algorithms (see \cref{sec:related}), we also assume access to earlier hypotheses in the form of oracles $O_{H^1},O_{H^2},\ldots,O_{H^{t-1}}$\footnote{For conciseness, we refer to $\mathrm{sign}(H^{t})$ and $-\mathrm{sign}(H^{t})$ as $H^{t}$ and $-H^{t}$ throughout this work. This notation can be interpreted as a confidence-weighted prediction.}. The pseudo-code for our algorithm is given in~\cref{alg:QAgBoost} which follows the classical \citet{kalai-kanade} algorithm~\footnote{For completeness, we give a short simplified analysis of the \citet{kalai-kanade} algorithm (henceforth referred to as the KK algorithm) in \cref{sec:kalai-kanade}.}. 

At a very high level, \cref{alg:QAgBoost} iteratively computes multiple hypotheses. To compute the hypothesis, say $H^t$ in iteration $t$, it first randomly relabels the examples $\{(x_i,y_i)\}_i$ in a careful manner, and then obtains a hypothesis $h^t$ from the relabeled examples. Next, it estimates the confidence margins $\alpha_t$ and $\beta_t$, and depending on their values, generates $H^t$ from $h^t$ and the hypothesis $H^{t-1}$ from the earlier iteration. The best among $\{H^1, H^2, \ldots\}$ is returned as the strong learner. We note a few key points. Firstly, the \textit{relabeling step} in the KK algorithm can be simulated as replacing each example $(x_i,y_i)$ with two conservatively weighted examples: $(x_i,y_i)$ with weight $(1+w^t_i)/2$ and $(x_i,\bar{y}_i)$ with weight $(1-w^t_i)/2$. We set $w^t_i=\textrm{min}\left\{1,e^{\{-H^{t-1}(x_i)\cdot y_i\}}\right\}$ which can be shown to be a conservative weighting function. It is easy to show using the Chernoff-Hoeffding bounds that $\alpha_t$ and $\beta_t$ are good estimates of $\cor{h^t}{\mathcal{D}^{\prime}_{w^t}}$ and $\cor{-H^{t-1}}{\mathcal{D}^{\prime}_{w^t}}$, respectively, when $m$ is large; here, $\mathcal{D}^{\prime}_{w^t}$ denotes the distribution $D$ relabeled by the weighting function ${w^t}$. 

The quantum algorithm essentially takes care of two things: creating an oracle to return a superposition of relabeled examples and estimating the confidence margins $\alpha_t$ and $\beta_t$. For the latter, we use quantum mean estimation with relative error (see \cref{sec:prelims}) to obtain estimates $\Tilde{\alpha}_t$ and $\Tilde{\beta}_t$ respectively. The first task is accomplished by performing standard operations on a superposition state. We now state the main theorem w.r.t. the complexity and correctness of \cref{alg:QAgBoost}, and provide further exposition and detailed proofs of \cref{alg:QAgBoost} in \cref{sec:qagboostdetails} and \cref{sec:proofs}.

\begin{restatable}[Quantum Agnostic Boosting]{theorem}{thmqagboost}
\label{thm:quantumagnosticboosting} 
Given a $(m,\kappa,\eta)$-weak quantum agnostic learner $A$ with a VC dimension of $d$, \cref{alg:QAgBoost} makes at most $\tildeO{\frac{1}{\eta^4\varepsilon^3}\sqrt{d}\log{\frac{1}{\delta}}}$ queries to $A$ and runs for an additional $\tildeO{\frac{n^2\cdot T}{\eta^2\varepsilon}{\sqrt{d}}\log{\frac{1}{\delta}}}$ time, to obtain a quantum $\left(\kappa/\eta\right)$ optimal agnostic learner with a probability of failure of at most $5\delta T$ for any $\varepsilon>0$ and $T=\bigO{\frac{1}{\eta^2\varepsilon^2}}$.
\end{restatable}

\begin{psketch}
Almost all steps follow the KK algorithm. The first major source of error arises from the estimation of the margins. So, let's focus on \cref{line:1_algo_boosting} where $H^t$ is generated. At this point, \cref{alg:QAgBoost} needs to determine the combined classifier for the next step. Accordingly, we pick a classifier among $h^t$ and $-H^{t-1}$ that best correlates to the optimal classifier in the relabeled distribution and add a weighted version of it to the earlier hypothesis $H^{t-1}$. We denote this classifier as $g^t$, and observe that it's corresponding estimated confidence margin is ${\gamma}_t=\mathrm{max}\left({\alpha}_t,{\beta}_t\right)$. Of course, the algorithm has computed only $\Tilde{\alpha}_t$ and $\Tilde{\beta}_t$, and we denote $\mathrm{max}\left(\Tilde{\alpha}_t,\Tilde{\beta}_t\right)$ as $\Tilde{\gamma}_t$. 

Observe that 
\begin{equation*}
    \begin{split}
        \abs{\Tilde{\gamma}^t - \cor{g^t}{\mathcal{D}^{\prime}_{w^t}}}&\leq\abs{\Tilde{\gamma}^t-\gamma^t }+ \abs{\gamma^t- \cor{g^t}{\mathcal{D}^{\prime}_{w^t}}}\\&\leq \epsilon\gamma^t+\nicefrac{\eta\varepsilon}{20}\\&\leq\nicefrac{\eta\varepsilon}{10}.
    \end{split}
\end{equation*}
The first inequality follows from the triangle inequality. The second inequality follows from relative estimation (\cref{lem:RelativeEstimation} in \cref{sec:prelims}) and Chernoff-Hoeffding bounds. The final inequality stems from observing that $\gamma^t\leq1$ and setting $\epsilon=\nicefrac{\eta\varepsilon}{20}$. This shows that $\Tilde{\gamma}_t$ is a good estimate of the correlation of $g^t$ on to the relabeled distribution. Therefore, \cref{alg:QAgBoost} chooses the right hypothesis with high probability in every iteration. Each iteration of \cref{alg:QAgBoost} makes $\tildeO{\frac{1}{\eta\varepsilon}\sqrt{m}\log{\frac{1}{\delta}}}$ queries for estimating various quantities using \cref{lem:RelativeEstimation}. This gives us the required query complexity. 
    
Finally, we note that there are three points of failure in every iteration of \cref{alg:QAgBoost}: (a) Estimation of $\Tilde{\gamma}_t$ fails w.p. $\leq 3\delta$, (b) weak learner $A$ fails to produce a hypothesis w.p. $\leq\delta$, and (c) estimating the correlation of $g^t$ fails w.p. $\leq\delta$. Therefore the entire algorithm fails w.p. at most $5\delta T$.
\end{psketch}

\section{Quantum Decision Tree Learning without Membership Queries}
\label{sec:QDTL}
This section shows how to obtain efficient decision tree learning algorithms in the agnostic setting without membership queries, in particular, using only states that are superpositions of pairs of random examples and labels provided by the \qaex oracle. We use an improper learning approach with two main steps: Obtain a weak learner and then use an appropriate boosting algorithm to obtain a strong learner (see \cref{fig:algo-overview}). Since the agnostic/adversarial noise setting is the hardest generalization of \pac learning, it follows that the above blueprint would also work for designing efficient learning algorithms for decision trees without \mq for more restricted noise models such as the random classification noise model, and the realizable/noiseless model. In fact, there exist simpler algorithms for both of these restricted settings as discussed in \cref{sec:QTDLreal} and \cref{sec:QTDLrandom}.
\subsection{The Adversarial Noise (Agnostic) setting}
\label{sec:QTDLagn} 
Here the task is to learn an unknown concept $f(x)$ (represented by a decision tree) given a \qaex oracle. 
There were several difficulties in using the existing techniques to construct a weak learner for the agnostic settings. For example, in the technique proposed by ~\citet{gopalan2008agnostically}, a function is implicitly constructed from the \textsc{Aex} oracle whose samples are used as an approximation of the true labeling function; this is, however, not possible due to the inherent differences between \textsc{Aex} and \qaex oracles.\vspace{0.1cm}

The approach taken in \cref{sec:QTDLreal} would also not work since it is not entirely clear how to obtain a Fourier sampling state\footnote{This is the state $\left(\frac{1}{2^{n/2}}\right)\sum_{x} \hat{f}(x)\ket{x}\ket{...}$.} directly from the \qaex oracle without an explicit oracle $O_f$ where $f(x)$ is the unknown concept we are trying to learn.

Finally, the techniques of \citet{biasedoraclerudyraymond} that we used for the random classification noise setting also do not apply here since it acts only with oracles $O^{\gamma}_f$ where the bias $\gamma$ is the same for all $x$. However, in the agnostic setting, the bias is dependent on $x$.

\begin{algorithm*}[!h]
    \caption{Weak Quantum Agnostic Learner}
    \label{alg:QWeakLearner}
    \DontPrintSemicolon
    \SetKwComment{Comment}{~$\vartriangleright$~}{}
    \LinesNumbered
    \footnotesize
    \SetKwInput{kwInit}{Initialize}
    \setcounter{AlgoLine}{0}
    \KwIn{The \qaex oracle, $t$. \textbf{Initialize} $\delta=\nicefrac{1}{10}$, $\gamma\leq\mathrm{min}\left\{\nicefrac{\delta}{4nt^2},\nicefrac{\varepsilon^2}{8}\right\}$.}
    \KwOut{A $(m,\nicefrac{1}{t},\varepsilon)$ WL for size-$t$ decision trees.}
    {Query the \qaex oracle to obtain $\ket{\psi_1}=\sum_{x,y}\sqrt{\mathcal{D}_{x,y}}\ket{x}\ket{y}=\frac{1}{\sqrt{2^n}}\sum_{x}\ket{x}\left(\sum_y\alpha_{y|x}\ket{y}\right)$.}\;
    {
    {Perform $\ell$ independent estimations of \mae($\varepsilon$,$1-\nicefrac{8}{\pi^2}$) conditioned on the second register to obtain 
    $\frac{1}{\sqrt{2^n}}\sum_x\ket{x}{\left(\sum_y\alpha_{y|x}\ket{y}\right)\left(\beta_{gx}\ket{\Tilde{\alpha}_{1|x}}+\beta_{bx}\ket{\mathrm{Err}}\right)^{\otimes \ell}}
    $.}
    \Comment{Let $\ell=\bigO{\log{1/\gamma}}$.}
    {On each of the $\ell$ registers, perform thresholding on 3rd register to obtain $\frac{1}{\sqrt{2^n}}\sum_x\ket{x}\left(\Hat{\beta}_{gx}\ket{h(x)}\ket{\phi^{\prime}(x)}+\Hat{\beta}_{bx}\ket{\overline{h(x)}}\ket{\phi^{\prime\prime}(x)}\right)^{\otimes \ell}$.} \Comment{Let $h(x)=\mathbb{I}\left[\Tilde{\alpha}_{1|x}>\frac{1}{\sqrt{2}}\right]$.}}
    {Perform majority on $\ket{h(x)}$ registers over all $\ell$ copies to create $\frac{1}{\sqrt{2^n}}\sum_x\ket{x} \ket{\xi(x)} \ket{h^*(x)}$.}\;
    {Let $O_h$ be the combined unitary from steps 1 to 4.}\label{line:weak-learner-Oh}\;
    {Perform a binary search over the intervals $(\tau', \tau]$ of size $\epsilon/16$ on $(0,1]$, to find the largest $\tau$ such that $\qgl(O_h, n, \tau, \epsilon, \delta/4\log(\varepsilon))$ outputs a tuple $(l, \tilde{S})$ with $l=1$. The search terminates if $\tau\le 1/t$.}\;
    {Return the parity monomial $\chi_{\Tilde{S}}$ as our weak learner.}\;
\end{algorithm*}

\cref{alg:QWeakLearner} constructs a weak quantum agnostic learner for decision trees from the $\qaex$ oracle\footnote{\cref{alg:QWeakLearner} is detailed in \cref{sec:QDTLagnapp}.}. We can then use \cref{alg:QAgBoost} to boost this weak learner into a quantum agnostic learner for decision trees. Algorithm 2 first constructs an operator $O_h$ using the $\qaex$ oracle that can act as a biased oracle for some predictor $f$ such that $f$ is an approximation of the Bayes optimal predictor$f_{\mathcal{B}}$. 

Internally, the algorithm checks for each $x$, in superposition, which amongst $\alpha_{0|x}$ and $\alpha_{1|x}$ is the largest and sets $h(x)$ accordingly. Further, it performs multiple such checks over multiple independent copies to reduce any error arising from the amplitude estimation of the $\alpha_{y|x}$ states. Next, \cref{alg:QWeakLearner} offloads the bulk of its work to a quantum rendering of the GL algorithm denoted $\qgl$\footnote{For details refer to \cref{alg:qgl} detailed in \cref{sec:qgl}.}.  

\qgl\ tries to approximate a decision tree with a monomial, and it's operations are motivated by the classical GL algorithm (see \cref{sec:prelims}). The technical difficulty was to generalize it to take as input a strongly biased oracle instead of an (error-free) oracle for a Boolean function and further enhance it to contain three kinds of errors: (a) errors from the biased oracle, (b) errors arising from amplitude estimation, and (c) errors from amplitude amplification (the state that we will amplify may contain false positives arising due to the first two errors, and those will now be incorrectly amplified). 

We now state the main theorem for obtaining Quantum Weak Learners, with a detailed proof in \cref{sec:QDTLagnappproofs}. 

\begin{restatable}[Weak Agnostic Learner for size-$t$ Decision Trees]{theorem}{thmwklearner}
\label{thm:weakagnosticdtlearner}
    Let $\eta=1/t$, and let $\kappa \in [0,1/2)$. Given access to a $\qaex$ oracle, \cref{alg:QWeakLearner} makes $m=\tildeO{\frac{n}{\eta\kappa^3}\cdot\log{\frac{1}{\kappa}}}$ calls to the \qaex oracle and runs for an additional  $\tildeO{\frac{n}{\eta\kappa^3}\cdot\log{\frac{1}{\kappa}}}$ time to obtain a $(m,\kappa,\eta)$-weak quantum agnostic learner for size-$t$ decision trees w.h.p.
\end{restatable}
\begin{psketch}
Let $\mathcal{C}$ be a family of size-$t$ decision trees with $c\in\mathcal{C}$ as the optimal classifier. Using the Fourier expansion of $c$ and applying \cref{def:correlation} we have $\cor{c(x)}{D}=\sum_{S\subseteq[n]}\hat{c}(S)\cor{\chi_S(x)}{D}$. \citeauthor{kushilevitz1991learning} showed that $\sum_{S\subseteq[n]}\abs{\hat{c}(S)}\leq t$. Using an averaging argument, we have $\mathrm{max}_S\abs{\cor{\chi_S(x)}{D}}\geq\frac{1}{t}\cor{c(x)}{D}$. 

We now claim that \cref{alg:QWeakLearner} produces $\Tilde{S}$ s.t. $\abs{\mathrm{max}_S\;\cor{\chi_{S}(x)}{D}-\cor{\chi_{\Tilde{S}}(x)}{D}}\leq \varepsilon$. This claim follows from the detailed analysis of the \qgl{} algorithm (\cref{alg:qgl}; see \cref{sec:qglproof} for details). Given $\Tilde{S}$, we have $\cor{\chi_{\Tilde{S}}(x)}{D}\geq \frac{1}{t}\cor{c(x)}{D}-\varepsilon$. This is an $\left(m,\varepsilon,\frac{1}{t}\right)$-weak quantum agnostic learner w.r.t $c$ (from \cref{def:corrAgnostic}). 
\end{psketch}

We state the main result of this work now.
\begin{theorem}[Restating \cref{thm:maininfo}]
    For any $\delta>0$, $\varepsilon\in(0,\nicefrac{1}{2})$, there exists a quantum learning algorithm with VC dimension $d$ that makes $\tildeO{\frac{nt^5\sqrt{d}}{\varepsilon^6}\log{(\nicefrac{1}{\delta})}}$ queries to the \qaex oracle and takes an additional $\tildeO{\frac{n^3t^5\sqrt{d}}{\varepsilon^6}\log{(\nicefrac{1}{\delta})}}$ time for $\left(t\varepsilon\right)$-optimal  agnostic PAC learning size-$t$ decision trees on $n$-bits.  
\end{theorem}
\begin{psketch}
    We use the weak quantum agnostic learner for size-$t$ decision trees constructed in \cref{alg:QWeakLearner} (set $\kappa=\epsilon$ and $\eta=\tfrac{1}{t}$ in \cref{thm:weakagnosticdtlearner}) as a weak learner for the quantum agnostic boosting algorithm as described in \cref{alg:QAgBoost}. By \cref{thm:quantumagnosticboosting}, the output of \cref{alg:QAgBoost} is a $(t\varepsilon)$-optimal agnostic learner for size-$t$ decision trees.
\end{psketch}

\begin{algorithm}[t!]
\caption{QGL algorithm}
\label{alg:qgl}
\setcounter{AlgoLine}{0}
    \DontPrintSemicolon
    \LinesNumbered
    \SetKwComment{Comment}{~$\vartriangleright$~}{}
    \footnotesize
    \SetKwInput{kwInit}{Initialize}
    \SetKwRepeat{Do}{do}{while}
    \KwIn{
      Oracle $\abpo$, n, threshold $\tau$, accuracy $\epsilon\in (0,\tau)$, error $\delta'\in (0,1/2)$}
    \KwOut{A tuple $(l,\Tilde{S})$ such that if $l=1$ then $\hat{h}(\Tilde{S})\ge \tau-\epsilon$ and if $l=0$ then $\hat{h}(S)<\tau,~\forall S$}
    \kwInit{Set $L_F$ to be the first level that has the number of nodes $r$ to be at least $\frac{1}{\tau^2}$ nodes. Set $\lig=L_F$.}
    \kwInit{Set $q=\lceil{\log(1/\epsilon)}\rceil+5$ and $\delta' = \delta\tau^2/8n$.}
    
        {Prepare the state $\ket{\psi_1}=\frac{1}{\sqrt{r}}\sum_{p'\in L_{i,g}}\ket{p'}$}\;
        
        \Do{$i\neq n+1$}{
            {Append $\ket{+}\ket{0}\ket{0^n}\ket{0^n}\ket{0^{q},0}^{\otimes O(\log(1/\delta'))}\ket{0}$ to the state.}\;
            {Prepare the state $\ket{\nu_1}$ in the fourth register and the state $\ket{\nu^p_2}$ in the fifth register.}\\\Comment{$\ket{\nu_1}$ and $\ket{\nu_2^p}$ are defined in \cref{sec:qglproof}.}
            {Perform the swap test with $3^{rd}$ register as the control state and $4^{th}$ and $5^{th}$ registers as the target state to obtain 
                $$\frac{1}{\sqrt{2|\lig|}}\sum_{p\in \lig\times \{0,1\}}\ket{P}\big(\ket{0^q}\ket{0}\big)^{\otimes\ell}\ket{0}.$$}\\\Comment{$L_{i,g}$ is the set of ``good'' prefixes of level $i-1$.}\Comment{$\ket{P}=\ket{p}(\sigma_{0,p}\ket{0}\ket{\phi_{0,p}} + \sigma_{1,p}\ket{1}\ket{\phi_{1,p}}).$}\Comment{\Comment{Let $\ell=\bigO{\log{1/\delta^{\prime}}}$.}}
                {Perform M.A.E.($\epsilon/2$, $1-\frac{8}{\pi^2}$) to estimate $\sigma_{0,p}$ and obtain a state of the form $$\frac{1}{\sqrt{2|\lig|}}\sum_{p\in \lig\times \{0,1\}}\ket{P}\ket{W}^{{\otimes\ell}}\ket{0}.$$ }\;\Comment{$\ket{W}=\big(\upsilon_{p,g}\ket{\widetilde{\sigma_{0,p}}} + \upsilon_{p,b}\ket{E_p}\big)\ket{0}.$}
            {In each of the $O(\log(1/\delta'))$ estimate registers, mark all the estimates that are at least $\frac{1}{2}+\frac{1}{2}(\tau-\epsilon)$.}\;
            {Perform a majority over all the $O(\log(1/\delta^{\prime}))$ indicator registers and store the majority in the last register.}\;
            
        {Amplify the probability of measuring the last register as $\ket{1}$ to obtain the following state of the form with error at most $\delta'/2n$
        $$\alpha_{i,g}\sum_{p\in L_{i+1,g}} \ket{p,\xi_p,1} + \alpha_{i,b}\sum_{p\in L_{i+1,b}} \ket{p,\xi_p,0}.$$
            }
            {Measure the last qubit as $m$. If $m=0$ and $i\neq n$, return to step $1$.}\;
        }
        {Measure the first register as $\tilde{S}$ and return $(m,\tilde{S})$}\;

\end{algorithm}
\subsection{The Noiseless (Realizable) Setting}
\label{sec:QTDLreal}
Many quantum algorithms use the Fourier sampling oracle to obtain speedups over their classical counterparts. A Fourier sampling oracle~\citep{BV97} yields the state $\sum_{S\subseteq[n]}\hat{f}(S)\ket{S}$, given access to an oracle for the function $f$, and upon measurement, returns $S$ such that $\hat{f}^2(S)$ is the largest with high probability. It is, therefore, natural to use as $f$ the labeling function in a realizable setting. 
Further, we can use the majority of several Fourier samples, from multiple copies of the above state, as a realizable weak learner for size-$t$ decision trees from $O_f$ without using membership queries (see \cref{sec:QDTLrealapp} for details). This weak learner can be fed into quantum realizable boosting algorithms~\cite{Arunachalam2020,deWolf2020} to obtain a strong PAC learner for size-$t$ decision trees.
\subsection{The Random Classification Noise setting}
\label{sec:QTDLrandom}

In this model, the labels associated with instances suffer from an independent random noise, and we can model it as a biased oracle $O^\epsilon_f$ for the true labeling function $f(x)$ s.t. $O^\epsilon_f\ket{x}\ket{0^{m-1}}\ket{0}$ gives us the state $\ket{x}\Big(\alpha\ket{u_x}\ket{f(x)}+\beta\ket{w_x}\ket{\overline{f(x)}}\Big)$ with $|\alpha|^2\ge \frac{1}{2}+\epsilon$. \citet{biasedoraclerudyraymond} showed that for any $O(T)$ query quantum algorithm that solves a problem with high probability using access to a perfect oracle, there exists an $O(T/\epsilon)$ query quantum algorithm that solves the same problem with high probability but using access to an $\epsilon$-biased oracle~\footnote{We provide a small discussion on \citet{biasedoraclerudyraymond} in \cref{sec:Iwama discussion} for completeness.}.Thus, to obtain a weak learner in the RCN setting, we only need to design a \qgl{} variant using an unbiased oracle, and then use the result by \citet{biasedoraclerudyraymond} to adapt it for a biased oracle. It suffices to state that the \qgl{} algorithm in \cref{alg:qgl} also works for unbiased oracles. 

\section{Discussion}
\citet{rudin2022interpretable} lists decision tree learning as one of ten grand challenges in interpretable machine learning. Current state-of-the-art decision tree learning algorithms (\cref{table:querycomplexity}) make use of membership queries that detract from human explainability. Therefore, there is a well-motivated need to move away from \mq and towards weaker query models. We give such an algorithm using \qaex queries in this work. We also remark here that the agnostic setting is particularly suitable for NISQ devices. However, since the boosting algorithms proposed in this work appear too complex to be implemented on NISQ hardware, simpler alternatives may be appealing, particularly to the practitioners of quantum ML. The ultimate goal is to obtain efficient learning algorithms for decision trees in the agnostic setting by only using random examples (from the training set). Another immediate follow-up would be obtaining lower bounds for improper learning of decision trees without \mq in the agnostic setting. 

\section{Acknowledgements}
The authors would like to thank Marcel Hinsche for pointing out an error in an earlier version of this work.

\printbibliography
\part{Appendix}

\section{The Kalai-Kanade Algorithm}\label{sec:kalai-kanade}
\begin{algorithm}[h!]
\setcounter{AlgoLine}{0}
    \setcounter{AlgoLine}{0}
    \DontPrintSemicolon
    \LinesNumbered
    \footnotesize
    \KwIn{
      $(m,\kappa,\eta)$-weak agnostic learner $A$ with complexity $R$, and $m$ labeled training samples $S=\{(x_i,y_i)\}_{i\in[m]}$.}
    \KwOut{$\left(\nicefrac{\kappa}{\eta}\right)$-optimal hypothesis $H^{\hat{t}}$ for $1\leq \hat{t}\leq T$ such that $\textrm{err}_S(H^{\hat{t}})=\textbf{argmin}_t \textrm{err}_S(H^{{t}})$.}
    \KwData{Initialize $H^0=0$, and a worst-case guess for $T$.}
        \For{$t=1$ to $T$}{
           
            {Define $w^t_i=-\phi^{\prime}(z_i)=\textrm{min}\left\{1,e^{\{-H^{t-1}(x_i)\cdot y_i\}}\right\}$.}\;
            
            {\textbf{Relabeling Step:} Set $\tilde{y}_i=y_i$ w.p. $(1+w^t_i)/2$, and w.p. $(1-w^t_i)/2$ set $\tilde{y}_i=\bar{y}_i$.}\;
            
            {Pass the set of relabeled samples $\tilde{S}=\{(x_i,\tilde{y}_i)\}_{i\in[m]}$ to $A$ to obtain intermediate hypothesis $h^t$.}\;
           
            {Let $\alpha_t=\frac{1}{{m}}\sum_{i\in[m]}\left(w^t_i\cdot y_i\cdot h^t(x_i)\right)$, and $\beta_t=\frac{1}{{m}}\sum_{i\in[m]}\left(w^t_i\cdot y_i\cdot -H^{t-1}(x_i)\right)$.}\;
            
            {If ${\alpha}_t>{\beta}_t$, set $H^t=H^{t-1}+{\alpha}_t\cdot h^t$. Otherwise, set $H^t=\left(1-{\beta}_t\right) H^{t-1}$.}\;
        }
\caption{The Kalai-Kanade algorithm}
\label{alg:PBAgBoost}
\end{algorithm}
%

We first define the \textit{conservative} weighting function used to relabel training samples.
\begin{defn}[Conservative weighting function]
\label{def:conserveweight}
A function $w:\mathcal{X}\times\{-1,1\}\xrightarrow{}[0,1]$ is conservative for any function $h:\mathcal{X}\xrightarrow{}\{-1,1\}$ if $w(x,-h(x))=1$ for all $x\in\mathcal{X}$.
\end{defn}
Consider the potential function $$\phi(z) = \begin{cases}1-z\;\;\text{if }z\leq 0\\ e^{-z}\;\;\;\;\text{if }z>0\end{cases}.$$
Observe that the weights in the Kalai-Kanade algorithm are set to the negative gradients of $\phi$ whose argument contains a combined hypothesis from the previous iterations; therefore, we try to use the weak learner to form a combined hypothesis that lowers the potential function in gradient descent like fashion. We note here that $\phi(z)$ is differentiable everywhere and $-\phi^{\prime}(z)\in\{1/e,1\}$. We state the following lemma using this fact and Taylor's expansion.
\begin{clm}[Lemma 2 of \cite{kalai-kanade}]
\label{cl:lem2kk}
${\phi(z)-\phi(z+\varepsilon)}\geq -\phi^{\prime}(z)\cdot\varepsilon-\frac{\varepsilon^2}{2}$.
\end{clm}

The Kalai-Kanade algorithm produces a combined classifier $H^t$ on round $t$, which has a lower potential than $H^{t-1}$ until the potential eventually drops from $1$ in iteration $t=1$ to (or gets arbitrarily close to) $0$ for some iteration $\hat{t}$. Since there is a lower bound on how much the potential can drop every round, this gives us an upper bound on the number of iterations until the Kalai-Kanade algorithm converges. Finally, we see that when the potential drops to its lowest value, the combined classifier $H^{\hat{t}}$ qualifies as an agnostic learner. Let $\mathcal{D}$ be any arbitrary joint distribution over $\mathcal{X}\times\{-1,1\}$. We denote the resulting relabeled distribution\footnote{Technically, this is $\mathcal{D}^{\prime}_{w^t}$, but the usage should be apparent from the context.} (relabeled using any weighting function $w:\mathcal{X}\times\{-1,1\}\xrightarrow{}[0,1]$) by $\mathcal{D}^{\prime}_w$.
\begin{clm}[Lemma 1 of \cite{kalai-kanade}]
\label{cl:lem1kk}
Given any arbitrary distribution $\mathcal{D}$ over $\mathcal{X}\times\{-1,1\}$, an optimal classifier $c$ and a classifier $h$ s.t. $c,h:\mathcal{X}\xrightarrow{}[-1,1]$, and a weighting function $w:\mathcal{X}\times\{-1,1\}\xrightarrow{}[0,1]$ which is conservative for $h$, we can show that $\cor{c}{\mathcal{D}^{\prime}_{w}}-\cor{h}{\mathcal{D}^{\prime}_{w}}\geq\cor{c}{\mathcal{D}}-\cor{h}{\mathcal{D}}$.
\end{clm}
\begin{proof}
From \cref{alg:PBAgBoost}, we can see that $\expectover{h(x)\cdot y}{\{x,y\}\in\mathcal{D}^{\prime}_{w}} = \expectover{h(x)\cdot y\cdot w(x,y)}{\{x,y\}\in\mathcal{D}}$. We now evaluate the quantity $\cor{c}{\mathcal{D}^{\prime}_{w}}-\cor{h}{\mathcal{D}^{\prime}_{w}}$ using \cref{def:correlation}. 
\begin{equation*}
    \begin{split}
        \cor{c}{\mathcal{D}^{\prime}_{w}}-\cor{h}{\mathcal{D}^{\prime}_{w}}&=\cor{c}{\mathcal{D}^{\prime}_{w}}-\cor{h}{\mathcal{D}^{\prime}_{w}}+\cor{c}{\mathcal{D}}-\cor{h}{\mathcal{D}}-\cor{c}{\mathcal{D}}+\cor{h}{\mathcal{D}}\\
        &=\cor{c}{\mathcal{D}}-\cor{h}{\mathcal{D}}+\cor{c}{\mathcal{D}^{\prime}_{w}}-\cor{c}{\mathcal{D}}-\cor{h}{\mathcal{D}^{\prime}_{w}}+\cor{h}{\mathcal{D}}\\
        &=\cor{c}{\mathcal{D}}-\cor{h}{\mathcal{D}}+\expectover{c(x)\cdot y\cdot \left(1-w(x,y)\right)}{\{x,y\}\in\mathcal{D}}-\expectover{h(x)\cdot y\cdot \left(1-w(x,y)\right)}{\{x,y\}\in\mathcal{D}}\\
        &=\cor{c}{\mathcal{D}}-\cor{h}{\mathcal{D}}-\expectover{\left(c(x)-h(x)\right)\cdot y\cdot \left(1-w(x,y)\right)}{\{x,y\}\in\mathcal{D}}.
    \end{split}
\end{equation*}

The proof follows from \cref{def:conserveweight}, and the fact that $w(x,y)=-\phi^{\prime}(h(x)y)$. When $h(x)=y$, we have $w(x,y)=1/e\implies 1-w(x,y)>0$, and $c(x)\cdot y\leq 1$ (true for any classifier). Therefore $\cor{c}{\mathcal{D}^{\prime}_{w}}-\cor{h}{\mathcal{D}^{\prime}_{w}}\geq\cor{c}{\mathcal{D}}-\cor{h}{\mathcal{D}}$. Alternatively, when $h(x)=-y$, we have $w(x,y)=1\implies 1-w(x,y)=0$ which implies $\cor{c}{\mathcal{D}^{\prime}_{w}}-\cor{h}{\mathcal{D}^{\prime}_{w}}=\cor{c}{\mathcal{D}}-\cor{h}{\mathcal{D}}$. 
\end{proof}

Consider the case when $\cor{\mathcal{C}}{\mathcal{D}^{\prime}_{w}}=0$. In this case, the optimal classifier behaves like a random guesser under the relabeled distribution. Therefore, either the combined classifier $H^{t-1}$ is worse than random guessing (since it was used to set the weights for relabeling), and we should use its negation as a weak agnostic learner, or the hypothesis returned by the weak learner trained on the relabeled distribution is close to optimal. Therefore, we need to pick either of these to add to the combined classifier for the next iteration. The selected hypothesis is denoted by $g^t$. The Kalai-Kanade algorithm combines the existing combined classifier and $g^t$ (weighted by its correlation $\gamma^t$) to form the combined classifier for the $t$\textsuperscript{th} iteration. Next we state a result that lower bounds the drop in potential in every iteration.
\begin{clm}[Lemma 3 of \cite{kalai-kanade}]
\label{cl:lem3kk}
Given any function $H:\mathcal{X}\xrightarrow{}\mathbb{R}$, hypothesis $h:\mathcal{X}\xrightarrow{}[-1,1]$, a weight $\gamma\in\mathbb{R}$, an arbitrary joint distribution $\mathcal{D}\sim\mathcal{X}\times\{-1,1\}$, a weighting function $w(x,y)=-\phi^{\prime}(y\cdot H(x))$, and a relabeled distribution $\mathcal{D}^{\prime}_{w}$, we have $\expectover{\phi(y\cdot H(x))}{\{x,y\}\sim\mathcal{D}}-\expectover{\phi\left(y\cdot \left(H+\gamma h\right)(x)\right)}{\{x,y\}\sim\mathcal{D}}\geq \cor{h}{\mathcal{D}^{\prime}_{w}} - \frac{\gamma^2}{2}$.
\end{clm}
The proof follows directly by plugging in appropriate values for $z$ and $\varepsilon$ in \cref{cl:lem2kk} and taking an expectation over both sides. 
The main result of \cite{kalai-kanade}, which shows that the combined classifier output by the Kalai-Kanade algorithm is an agnostic learner, is as follows. 
\begin{restatable}
[Theorem 1 of~\cite{kalai-kanade}]{lem}{lemagnosticboosting}

Let $A$ be an $(m,\kappa,\eta)$-weak agnostic learner w.r.t. some concept class $\mathcal{C}$ s.t. $\mathrm{VCdim}(\mathcal{C})=d$. Then, for any $\varepsilon,\delta>0$, there exists an agnostic boosting algorithm that uses $m=\bigO{\frac{1}{\eta^2\varepsilon^2}\log{\frac{1}{\delta}}}$ examples and $T=\bigO{\frac{1}{\eta^2\varepsilon^2}}$ iterations, makes $\tildeO{\frac{T\cdot d}{\eta^2\varepsilon}\log{\frac{1}{\delta}}}$ queries to $A$ and runs for an additional $\tildeO{\frac{n^2\cdot T\cdot d}{\eta^2\varepsilon}\log{\frac{1}{\delta}}}$ to output a hypothesis $h$ with probability at least $1=\bigO{\delta T}$, such that
$
\cor{h}{\mathcal{D}}\geq\mathrm{optcor}_{\mathcal{D}}(\mathcal{C})-\frac{\kappa}{\eta}-\varepsilon
$.\label{lem:agnosticboosting}
\end{restatable}

For a large enough training set size, we can give a tight enough estimate for the correlation of the new classifier $g^t$, which is an $(m,\kappa,\eta)$-weak agnostic learner. We also see from \cref{cl:lem3kk} that a confidence-based weighted combination drops the potential, and we can lower bound this drop in potential. Therefore, we can obtain an upper bound on the number of iterations of \cref{alg:PBAgBoost}, such that the potential function eventually reaches the minimum possible value. The proof follows from the fact that when the potential function reaches the minimum possible value, the corresponding combined classifier is a $\kappa/\eta$-optimal agnostic learner. 

\subsection{Proof of \cref{lem:agnosticboosting}}
\label{sec:proofagnosticboosting}
\begin{clm}
Either the weak hypothesis produced by \cref{alg:PBAgBoost} on the $t$\textsuperscript{th} iteration, or the negation of the combined hypotheses up to the $t-1$\textsuperscript{th} step has a correlation greater than $\frac{\eta\varepsilon}{3}$.\label{cl:eitheror}
\end{clm}
\begin{proof}
Consider the optimal hypothesis $c\in\mathcal{C}$, and the combined hypothesis produced by \cref{alg:PBAgBoost} at iteration $t-1$ to be $H^{t-1}$. If $H^{t-1}$ is not a $\beta$-optimal agnostic learner, then we have $\cor{c}{\mathcal{D}}>\cor{H^{t-1}}{\mathcal{D}}+\beta+\varepsilon$. Plugging in \cref{cl:lem1kk}, we have $\cor{c}{\mathcal{D}^{\prime}_w}>\cor{H^{t-1}}{\mathcal{D}^{\prime}_w}+\beta+\varepsilon$.

First consider the case $\cor{c}{\mathcal{D}^{\prime}_w}>\beta+\frac{\varepsilon}{2}$, where $\beta=\frac{\kappa}{\eta}$. Consider the hypothesis $h^t$ produced by the weak learner at the $t$\textsuperscript{th} iteration in \cref{alg:PBAgBoost}. By the weak learning assumption, we have $\cor{h^t}{\mathcal{D}^{\prime}_w}\geq\eta\cdot\cor{c}{\mathcal{D}^{\prime}_w}-\kappa\implies\cor{h^t}{\mathcal{D}^{\prime}_w}\geq \frac{\eta\varepsilon}{2}$. 

Now consider the other case $\beta+\frac{\varepsilon}{2}>\cor{c}{\mathcal{D}^{\prime}_w}$. This implies that $\cor{-H^{t-1}}{\mathcal{D}^{\prime}_w}>\frac{\varepsilon}{2}$.
\end{proof}
\lemagnosticboosting*
\begin{proof}
Since $\eta\in\left[0,\frac{1}{2}\right)$, we have from \cref{cl:eitheror} that $\cor{g^t}{\mathcal{D}^{\prime}_w}\geq\frac{\eta\varepsilon}{3}$, where $g^t$ is the better of the two candidate hypotheses at iteration $t$. Now, consider the margin $\gamma^t$ of the best classifier $g^t$ at iteration $t$ obtained using $m$ training samples.
$$
\gamma^t=\frac{1}{m}\sum_{i\in[m]}g^t(x_i)\cdot y_i\cdot w^t(x_i,y_i).
$$
This margin is simply the estimated correlation of $g^t$. Using Chernoff-Hoeffding bounds and setting $m=\bigO{\frac{1}{\eta^2\varepsilon^2}\log{\frac{1}{\delta}}}$, we have $\abs{\gamma^t-\cor{g^t}{\mathcal{D}^{\prime}_w}}\leq\bigO{\eta\varepsilon}$ with high probability. Setting the appropriate values for $\cor{g^t}{\mathcal{D}^{\prime}_w}$ allows us to lower bound the potential drop to at least $\bigO{\eta^2\varepsilon^2}$ in iteration $t>0$ using \cref{cl:lem3kk}. 

Since the potential function is bounded in the range $[0,1]$, and the potential drops by at least $\bigO{\eta^2\varepsilon^2}$, in $\bigO{\frac{1}{\eta^2\varepsilon^2}}$ iterations, \cref{alg:PBAgBoost} must produce a hypothesis such that the potential function drops to its lowest value. 
Consider the iteration $\tau$ in which potential drops to its lowest. From \cref{cl:lem1kk} we have
\begin{equation*}
    \begin{split}
        \cor{c}{\mathcal{D}^{\prime}_{w}}-\cor{g^\tau}{\mathcal{D}^{\prime}_{w}}&\geq\cor{c}{\mathcal{D}}-\cor{g^\tau}{\mathcal{D}}\\
        \implies\cor{g^\tau}{\mathcal{D}}\geq\cor{c}{\mathcal{D}}-&\left[\cor{c}{\mathcal{D}^{\prime}_{w}}-\cor{g^\tau}{\mathcal{D}^{\prime}_{w}}\right].
    \end{split}
\end{equation*}
Substituting $\cor{c}{\mathcal{D}^{\prime}_w}>\frac{\kappa}{\eta}+\frac{\varepsilon}{2}$ (since the potential is lowest at this iteration) and $\cor{g^\tau}{\mathcal{D}^{\prime}_{w}}\geq\frac{\eta\varepsilon}{3}$, we have $\cor{H^\tau}{\mathcal{D}}\geq\cor{c}{\mathcal{D}}-\frac{\kappa}{\eta}-\varepsilon$. Therefore, we have that  in $\bigO{\frac{1}{\eta^2\varepsilon^2}}$ iterations, \cref{alg:PBAgBoost} produces a $\left(\frac{\kappa}{\eta}\right)$-optimal agnostic learner.
\end{proof}

\bigskip
\section{Details of Quantum Agnostic Boosting Algorithm (\cref{alg:QAgBoost})}
\label{sec:qagboostdetails}
Prepare a set $S$ of $m$ training samples $\{(x_i,y_i)\}_{i\in[m]}$ by measuring the output of \qaex.
At the start of every iteration, we prepare $2+m$ copies of the uniform state 
\begin{equation*}
    \ket{\psi_0}=\ket{\phi_0}=\frac{1}{\sqrt{m}} \sum_{i \in [m]} \ket{x_i,y_i}.
\end{equation*}
Then, we query the $t-1$\textsuperscript{th} oracle $O_{H^{t-1}}$.
\begin{align*}
    \frac{1}{\sqrt{m}} \sum_{i \in [m]} \ket{x_i,y_i}\ket{0}\ket{0}
    \xrightarrow{O_{H^{t-1}}} &\frac{1}{\sqrt{m}}\sum_{i\in[m]}\ket{x_i,y_i}\underset{z_i}{\ket{\underbrace{-H^{t-1}(x_i)\cdot y_i}}}\ket{0}\\
    \xrightarrow{}&\frac{1}{\sqrt{m}}\sum_{i\in[m]}\ket{x_i,y_i}\ket{z_i}{\ket{w^t_i}}.
\end{align*}
The second step uses arithmetic operations to compute $w^t_i=\textrm{min}\left\{1,e^{\{-H^{t-1}(x_i)\cdot y_i}\}\right\}$.
We uncompute the $\ket{z_i}$ register using one query to the $O_{H^{t-1}}$ oracle to obtain $2+m$ copies of the state
\begin{equation*}
    \ket{\phi_2}=\frac{1}{\sqrt{m}}\sum_{i\in[m]}\ket{x_i,y_i}\ket{w^t_i}.
\end{equation*}
Take the first $m$ copies of $\ket{\phi_2}$, and perform arithmetic operations to obtain $m$ copies of the state
\begin{equation*}
    \frac{1}{\sqrt{m}}\sum_{i\in[m]}\ket{x_i,y_i}\ket{\frac{1+w^t_i}{2}}.
\end{equation*}
Perform a conditional rotation on the third register to obtain the state $\ket{\phi_3}$ as shown in \cref{line:phi3}.
\begin{equation*}
    \ket{\phi_3}=\frac{1}{\sqrt{m}}\sum_{i\in[m]}\ket{x_i,y_i}\left(\sqrt{\frac{1+w^t_i}{2}}\ket{0}+\sqrt{\frac{1-w^t_i}{2}}\ket{1}\right).
\end{equation*}
After we perform the C-NOT, we get $Q$ copies of a state $\ket{\phi_4}$ with \textit{conservatively} relabeled samples, as shown in \cref{line:phi4}.
\begin{equation*}
    \ket{\phi_4}=\frac{1}{\sqrt{m}}\sum_{i\in[m]}\ket{x_i}\left(\sqrt{\frac{1+w^t_i}{2}}\ket{y_i,0}+\sqrt{\frac{1-w^t_i}{2}}\ket{\bar{y}_i,1}\right).
\end{equation*}

We denote the unitary for obtaining $\ket{\phi_4}$ as $\textsc{Qaex}_{t}$. Now, we pass $\textsc{Qaex}_{t}$ to the $(m,\kappa,\eta)$-weak quantum agnostic learner $A$, to obtain query access to the $t$\textsuperscript{th} intermediate hypothesis $h^t$. 
 Note that the weak learner $A$ obtains the intermediate hypothesis using $\textsc{Qaex}_t$ as the quantum example oracle instead of \qaex.

At this point, we have two copies of $\ket{\phi_2}$ left over. On the first copy, use the $O_{h^t}$ oracle to obtain
\begin{equation*}
    \ket{\psi_3^1} = \frac{1}{\sqrt{m}}\sum_{i\in[m]}\ket{x_i,y_i}\ket{w^t_i}\ket{w^t_i\cdot y_i\cdot h^t(x_i)}.
\end{equation*}
Perform a conditional rotation on the last register to obtain
\begin{align*}
    \ket{\psi_4^1} &= \frac{1}{\sqrt{m}}\sum_{i\in[m]}\sqrt{\kappa_i}\ket{x_i,y_i}\ket{w^t_i}\ket{\kappa_i}\ket{1}\\
    &\hspace{1.2cm}+\frac{1}{\sqrt{m}}\sum_{i\in[m]}\sqrt{1-\kappa_i}\ket{x_i,y_i}\ket{w^t_i}\ket{\kappa_i}\ket{0}
\end{align*}
where ${\kappa_i}={w^t_i\cdot y_i\cdot h^t(x_i)}$. We can rewrite the first part as
\begin{equation*}
    \sqrt{\alpha_t}\sum_{i\in[m]}\sqrt{\frac{\kappa_i}{{\sum_{i\in[m]}\kappa_i}}}\ket{x_i,y_i}\ket{w^t_i,\kappa_i,1}.
\end{equation*}
We perform quantum amplitude estimation with relative error $\varepsilon$, conditioned on the $\ket{1}$ register, to obtain an estimate $\tilde{\alpha}_t$. On the second copy, use the $O_{H^{t-1}}$ oracle to obtain the state 
\begin{equation*}
    \ket{\psi_3^2} = \frac{1}{\sqrt{m}}\sum_{i\in[m]}\ket{x_i,y_i}\ket{w^t_i}\ket{w^t_i\cdot y_i\cdot -H^{t-1}(x_i)}.
\end{equation*}
Let ${\kappa_i}={w^t_i\cdot y_i\cdot -H^{t-1}(x_i)}$. Perform a conditional rotation on the last register to obtain the state
\begin{align*}
    \ket{\psi_4^2} &= \frac{1}{\sqrt{m}}\sum_{i\in[m]}\sqrt{\kappa_i}\ket{x_i,y_i}\ket{w^t_i}\ket{\kappa_i}\ket{1}\\
    &\hspace{1.2cm}+\frac{1}{\sqrt{m}}\sum_{i\in[m]}\sqrt{1-\kappa_i}\ket{x_i,y_i}\ket{w^t_i}\ket{\kappa_i}\ket{0}.
\end{align*}
We can rewrite the first part as $$\sqrt{\beta_t}\sum_{i\in[m]}\sqrt{\frac{\kappa_i}{{\sum_{i\in[m]}\kappa_i}}}\ket{x_i,y_i}\ket{w^t_i,\kappa_i,1}.$$ Again, we perform quantum amplitude estimation with relative error 
$\varepsilon$ to obtain an estimate for $\tilde{\beta}_t$. We now state the following claims.
\begin{restatable}{clm}{clmtriangle}
\label{cl:triangle}
\cref{alg:QAgBoost} computes estimates of margins $\Tilde{\alpha}_t$ and $\Tilde{\beta}_t$ s.t. $\Tilde{\gamma}_t=\mathrm{max}\left(\Tilde{\alpha}_t,\Tilde{\beta}_t\right)$ using $\tildeO{\frac{1}{\eta\varepsilon}\sqrt{m}\log{\frac{1}{\delta}}}$ queries. $\abs{\Tilde{\gamma}^t - \cor{g^t}{\mathcal{D}^{\prime}_{w^t}}}\leq\nicefrac{\eta\varepsilon}{10}$ with probability $\geq 1-3\delta T$.
\end{restatable}
\noindent\cref{cl:triangle} shows that we can estimate the correlation of the best classifier $g^t$ at every step $t>0$ with a high probability. 
\begin{restatable}{clm}{clmprobab}
\label{cl:probab}
\cref{alg:QAgBoost} takes as input an $(m,\kappa,\eta)$-weak quantum agnostic learner and outputs a $\left(\kappa/\eta\right)$-quantum agnostic learner with a probability of failure of at most $5\delta T$.
\end{restatable}
\noindent\cref{cl:probab} shows that our algorithm succeeds with high probability.
\begin{restatable}{clm}{clmquery}
\label{cl:querycomplexity}
Given a weak $(m,\kappa,\eta)$-weak quantum agnostic learner $A$ with a VC dimension of $d$, \cref{alg:QAgBoost} makes at most $\tildeO{\frac{1}{\eta^4\varepsilon^3}\sqrt{d}\log{\frac{1}{\delta}}}$ queries to $A$.
\end{restatable}
\noindent\cref{cl:querycomplexity} gives an upper bound on the query complexity of our boosting algorithm. 

Combining the three claims, we get \cref{thm:quantumagnosticboosting}, which states that the hypothesis $h$ produced by our agnostic boosting algorithm is very close to the accuracy of the best hypothesis in the concept class $\mathcal{C}$ with high probability, essentially guaranteeing that our boosting algorithm agnostically learns $\mathcal{C}$. All the proofs are given in \cref{sec:proofs}.
\section{Analysis of \cref{alg:QAgBoost}}
\label{sec:proofs}
The analysis of \cref{alg:QAgBoost} relies heavily on the analysis of the classical Kalai-Kanade algorithm as presented in \cref{sec:kalai-kanade} and \cref{sec:proofagnosticboosting}.
\subsection{Proof of Correctness}
The following claim shows us that the estimated quantity $\tilde{\gamma}_t$ in every iteration of \cref{alg:QAgBoost} is good.
\clmtriangle*
\begin{proof}
Let $g^t$ be the classifier chosen by \cref{alg:QAgBoost} at the $t$\textsuperscript{th} iteration. We denote the correlation of $g^t$ w.r.t. the relabeled distribution as $\cor{g^t}{\mathcal{D}^{\prime}_w}$. Using \cref{def:correlation}, we can restate this as
\begin{equation}
    \cor{g^t}{\mathcal{D}^{\prime}_{w^t}}=\expectover{w^t_i\cdot y_i\cdot g^t(x_i)}{x_i,y_i\sim\mathcal{D}}.
\end{equation}
Let $X_i=w^t_i\cdot y_i\cdot g^t(x_i)$ be a random variable. Applying \cref{def:conserveweight}, we get that $X_i\in\left[-\frac{1}{e},1\right]$. Let $\gamma^t=\frac{1}{m}\sum_{i\in[m]}X_i$.
Then by applying Chernoff-Hoeffding bounds, we have
\begin{align*}
    &{\mathrm{Pr}}\left[\abs{\gamma^t-\cor{g^t}{\mathcal{D}^{\prime}_{w^t}}}\geq\frac{\eta\varepsilon}{20}\right]\\&\leq 2\cdot\mathrm{exp}\left(\frac{-2\frac{\eta^2\varepsilon^2}{400}}{\sum_{i=1}^m \left(1+\frac{1}{e}\right)^2}\right)\\
    &\leq 2\delta.
\end{align*}

Therefore by setting $m={\frac{200}{\eta^2\varepsilon^2}\log{\frac{1}{\delta}}}$, we can obtain with probability at least $1-2\delta$,
\begin{equation}
    \label{eq:tr1}
 \abs{{\gamma}^t - \cor{g^t}{\mathcal{D}^{\prime}_w}}\leq\frac{\eta\varepsilon}{20}.  
\end{equation}
We can obtain an estimate $\Tilde{\gamma}^t$ of $\gamma^t$ using \cref{lem:RelativeEstimation} with probability at least $1-\delta$, such that 
\begin{equation}
    \label{eq:tr2}
    \abs{\Tilde{\gamma}^t-\gamma^t}\leq \epsilon\cdot\gamma^t.
\end{equation}
We note here that \cref{eq:tr2} and \cref{cl:eitheror} together make it impossible for the estimate $\Tilde{\gamma}^t$ to be so far from the actual margin $\gamma^t$, that we end up choosing the classifier with the worse correlation.

Use triangle inequality on \cref{eq:tr1} and \cref{eq:tr2}, we obtain with probability at least $1-3\delta$,
\begin{equation}
    \begin{split}
        \abs{\Tilde{\gamma}^t - \cor{g^t}{\mathcal{D}^{\prime}_{w^t}}}&\leq\abs{\Tilde{\gamma}^t-\gamma^t + \gamma^t- \cor{g^t}{\mathcal{D}^{\prime}_{w^t}}}\\
        &\leq\abs{\Tilde{\gamma}^t-\gamma^t} + \abs{\gamma^t- \cor{g^t}{\mathcal{D}^{\prime}_{w^t}}}\\
        &\leq \epsilon\cdot\gamma^t + \frac{\eta\varepsilon}{20}.
    \end{split}
\end{equation}
 In the last step we observe that $\gamma^t$ can be at most 1. Setting $\epsilon=\frac{\eta\varepsilon}{20}$ gives us the required upper-bound on  $\abs{\Tilde{\gamma}^t - \cor{g^t}{\mathcal{D}^{\prime}_{w^t}}}$. As a final point, we get the required query complexity by plugging the terms of \cref{eq:tr2} into \cref{lem:RelativeEstimation}. 
\end{proof}
We now show that our boosting algorithm actually boosts the given weak learner to produce an agnostic learner.
\clmprobab*
\begin{proof}
Using \cref{cl:triangle} and \cref{cl:lem3kk}, we obtain that the drop in potential for \cref{alg:QAgBoost} at every iteration is bounded by at most $\bigO{{\eta^2\varepsilon^2}}$. We now follow the proof for \cref{lem:agnosticboosting} given in \cref{sec:proofagnosticboosting} to show that \cref{alg:QAgBoost} produces a $(\kappa/\eta)$- agnostic learner in at most $\bigO{\frac{1}{\eta^2\varepsilon^2}}$ iterations. 

We allow the algorithm to fail with probability $3\delta$ during estimation of $\tilde{\gamma}^t$ (see \cref{cl:triangle}). We allow the algorithm to fail with another $\delta$ probability while invoking the weak learner to produce a hypothesis $h^t$ at the $t$\textsuperscript{th} iteration. Finally, estimating the correlation of the constructed hypothesis $g^t$ can fail with an additional probability of $\delta$ at every iteration.
\end{proof}
\subsection{Complexity Analysis}
\label{sec:complexity}
\clmquery*
\begin{proof}
The quantum algorithm runs for $\bigO{\frac{1}{\eta^2\varepsilon^2}}$ iterations (see \cref{cl:probab}). From \cref{cl:triangle}, we see that each iteration makes $\tildeO{\frac{\sqrt{m}}{\eta\varepsilon}\log{\frac{1}{\delta}}}$ queries. 
Plugging in sample complexity upper bounds from \cite{arunachalam2018optimal}, we have $m=\tildeTheta{\frac{d}{\eta^2}}$ for both the classical and quantum case\footnote{Refer Theorem 14 of \cite{arunachalam2018optimal} for the optimal quantum agnostic sample complexity.}, where $d$ is the VC-dimension of the $\left(\frac{\kappa}{\eta}\right)$-optimal agnostic learner. This gives us a total of $\tildeO{\frac{\sqrt{d}}{\eta^4\varepsilon^3}\log{\frac{1}{\delta}}}$ queries made by \cref{alg:QAgBoost}.
\end{proof}
We note here that the classical algorithm has a query complexity of $\tildeO{\frac{d}{\varepsilon^2}\log{\frac{1}{\delta}}}$~\citep{arunachalam2018optimal}. Therefore, we have a polynomial blowup in the given parameters, while we have a quadratic speedup in the VC dimension of the agnostic learner. We restate the main theorem here for completeness.

\thmqagboost*
\section{Quantum Goldreich-Levin Algorithm}
\label{sec:qgl}
\begin{figure}[!h]
    \centering
    \includegraphics[width=0.5\textwidth]{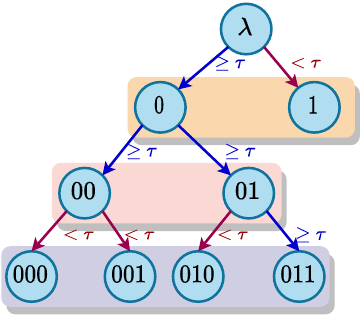}
    \caption{\footnotesize A partial \qgl{} tree (up to $3$ levels) indicating the level ordered traversal of ``good" prefixes, that have \pwc{} above threshold $\tau$. Bad prefixes are indicated using red arrows. The sub-trees of bad prefixes are not explored further. The set of good prefixes for level $i+1$ is decided in superposition over the set of good prefixes of level $i$. Shaded boxes indicate nodes evaluated in superposition.}
    \label{fig:qgltree}
\end{figure}

\begin{restatable}{clm}{clmqgl}
Given an oracle $O_h$, threshold $\tau$, accuracy $\epsilon$ and error parameter $\delta$, \cref{alg:qgl}  performs $O(\frac{n}{\epsilon^2\tau}\log(\frac{\delta\tau^2}{n}))$ queries to $O_h$ and outputs a pair $(l,\tilde{S})$ such that if $l=1$, then $\hat{h}(\tilde{S})\ge \tau-\epsilon$, else if $l=0$, then $\nexists\, S$ such that $\hat{h}(S)\ge \tau$, both w.p. $\geq 1-\delta$.
\label{clm:qgl}
\end{restatable}
\begin{algorithm}[h!]
	\caption{\iglalgo: Interval-Search Goldreich Levin \label{algo:iglalgo}~\citep{bera2021quantumnonlinear}}
\DontPrintSemicolon
    \LinesNumbered
    \footnotesize
    \setcounter{AlgoLine}{0}
\SetKwComment{Comment}{~$\vartriangleright$~}{}
	    \KwIn{Oracle $\abpo$, n, accuracy $\epsilon$ and probability of error $\delta$.}
	    {Set $k = \left\lceil \log_2 \tfrac{1}{\epsilon}
	    \right\rceil + 1$ \Comment{$k$ is the smallest integer {\it s.t.}
	    $\tfrac{1}{2^k} \le \frac{\epsilon}{2}$; thus, $\frac{\epsilon}{4}
	    < \frac{1}{2^k} \le \frac{\epsilon}{2}$}}
	    {Set gap $g=\frac{1}{8}\left(\epsilon - \tfrac{1}{2^k}\right)$ and threshold
	    $\tau=\frac{1}{2}$.
	    \Comment{$8g + \frac{1}{2^k} = \epsilon \implies \tfrac{3}{2}\tfrac{\epsilon}{16} \ge g \ge \tfrac{\epsilon}{16}$}
	    }
     \For{$i=1 \ldots k$}{
        {Invoke {\qgl}$(\abpo, n, \tau, g, \frac{\delta}{k})  \to (l, \tilde{S})$.}\;
		{\textbf{If} {$l=0$}, set $\tau = \tau + \tfrac{1}{2^{i+1}}$. \textbf{Else} set $\tau = \tau -\tfrac{1}{2^{i+1}}$.}\;
        {\textbf{If} {$\frac{1}{2^{i+1}}<g/2$ or $\tau \le \frac{1}{t}$}
            {return $\tilde{S}$.}}\;
    }
   \end{algorithm}

\subsection{Proof of correctness of Algorithm~\ref{alg:qgl}:}
\label{sec:qglproof}
We first present how the state evolves at each level of the quantum Goldreich-Levin algorithm.
Consider the $i^{th}$ level. 
Let $L_i$ denote the $i^{th}$ level of the Goldreich-Levin tree.
Also, let $L_{i,g}$ be the set of ``good" prefixes of level $i-1$.
(By ``good," we mean the prefixes $p$ such that $\pwc(p)$ is greater than the threshold.)
The state obtained at the end of the $i-1^{th}$ level will be of the form
\begin{equation*}
    \ket{\psi_{i}}=\begin{cases}
        \frac{1}{|L_{i,g}|}\sum_{p'\in L_{i,g}}\ket{p'}\ket{\phi_{p'}},&i\neq F\\
        &\\
        \frac{1}{\sqrt{r}}\sum_{p'\in L_F}\ket{p'},&i=F.
    \end{cases}
\end{equation*}

Let $k=\bigO{\log(1/\delta')}$. We append the state $\ket{+}\ket{0}^{1+2n}(\ket{0^q}\ket{0})^{\otimes{k}}\ket{0}$ to $\ket{\psi_i}$ to get
\begin{align*}
    \ket{\psi_i^{(1)}}&= {\frac{1}{\sqrt{|\lig|}}\sum_{p'\in \lig}}\ket{p'}\ket{+}\ket{\phi_{p'}}\ket{0}\ket{0^n}\ket{0^n}\Big(\ket{0^q}\ket{0}\Big)^{\otimes k}\ket{0}\\
    &={\frac{1}{\sqrt{2|\lig|}}\sum_{p'\in \lig}}\ket{p'^\frown 0}\ket{\phi_{p'}}\ket{0}\ket{\rho}\ket{0^n}\ket{0^n}(\ket{0^q}\ket{0})^{\otimes k}\ket{0}\\
    &+\ket{p'^\frown 1} \ket{\phi_{p'}}\ket{0}\ket{\rho}\ket{0^n}\ket{0^n}(\ket{0^q}\ket{0})^{\otimes k}\ket{0}\\
    &={\frac{1}{\sqrt{2|\lig|}}\sum_{p\in \lig\times \{0,1\}}}\ket{p}\ket{\phi_p}\ket{0}\ket{0^n}\ket{0^n}(\ket{0^q}\ket{0})^{\otimes k}\ket{0}\\
    &= R_1R_2R_3R_4R_{5}R_{6}R_7R_8~~~(\text{say})
\end{align*}

where $p'^\frown 0$ and $p'^\frown 1$ are $p'$ concatenated with $0$ and $1$ respectively, $\ket{+} = \nicefrac{(\ket{0}+\ket{1})}{\sqrt{2}}$, $R_6 R_7 = R_{6,1}R_{7,1}\cdots R_{6,k}R_{7,k}$.

Notice that the first register contains an equal superposition of all the immediate children of the ``good" prefixes of the previous level.
In the next step, we prepare the state $\ket{\nu_1}$ in $R_4$ where 
\begin{align*}
    \ket{\nu_1} &= \frac{1}{2^{n/2}}\sum_x \ket{x}\Big[\eta_{g,x}(-1)^{h(x)}\ket{h(x)}\ket{\psi_{x,g}}+\eta_{x,b}(-1)^{\ovl{h(x)}}\ket{\ovl{h(x)}}\ket{\psi_{x,b}}\Big].
\end{align*}
We also prepare the state $\ket{\nu_2^p}$ in $R_5$ where
\begin{align*}
    \ket{\nu_2^p} &= \frac{1}{2^{n/2}}\sum_x \ket{x}(-1)^{x_1\cdot p}\Big[\eta_{g,x}\ket{h(x)}\ket{\psi_{x,g}}+\eta_{x,b}\ket{\ovl{h(x)}}\ket{\psi_{x,b}}\Big].
\end{align*}
Then, we perform the swap test with $R_3$ as the control qubit and $R_4$ and $R_5$ as the target qubits.
This gives us,
$$\sigma_{0,p}\ket{0}\ket{\phi_{0,p}} + \sigma_{1,p}\ket{1}\ket{\phi_{1,p}}$$
as the state of the registers $R_3, R_4$ and $R_5$ for each $p$ where $|\sigma_{0,p}|^2 = \frac{1}{2}+\frac{1}{2}|\braket{\nu_1}{\nu_2^p}|^2$.

Next, for each $j=1,\cdots,O(\log(1/\delta'))$, we use M.A.E to $\epsilon/2$-estimate $|\sigma_{0,p}|^2$ in $R_{6,i}$ with error at most $1-\frac{8}{\pi^2}$ and flip the state in $R_{7,i}$ to $\ket{1}$ if the estimate is at least $\frac{1}{2}-\frac{1}{2}(\tau-\epsilon)$.
Notice that this essentially marks all the ``good" states but with an error $1-\frac{8}{\pi^2}$, i.e., the algorithm acts as a biased oracle to mark the ``good" states.

As the next step, we perform a majority over $O(\log(1/\delta'))$ $R_{7,i}$ copies and store the result in $R_8$.
This is followed by an amplitude amplification to obtain the ``good" states with high probability.
For the correctness of majority followed by amplitude amplification, we direct the reader to Appendix~H of~\cite{bera2022few}.
As the last step for this level, we measure $R_8$. If the measurement outcome is $1$, then the post-measurement state would contain an equal superposition of all the ``good" prefixes of that level.

Now, we analyze the quality of the estimate returned by the algorithm.
Recall that the sum of the squares of the Fourier coefficients of a function $h$ at all points with prefix $p$ can be given as $\pwc_h(p) $

\begin{align*}
    &= \sum_{s\in \{0,1\}^{n-|p|}}\hat{h}^2(p^\frown s)\\
    &= \mathbb{E}_{X_1,X_2,Z_1,Z_2}\big[(-1)^{f(X_1^\frown X_2)\xor f(Z_1^\frown Z_2)\xor p\cdot X_1\xor p\cdot Z_1}\big]\\
    &= \frac{1}{2^{2n}}\sum_{x_1,x_2,z_1,z_2}(-1)^{f(x_1^\frown x_2)\xor f(z_1^\frown z_2)\xor p\cdot x_1\xor p\cdot z_1}
\end{align*}

where the random variables $X_1$ and $Z_1$ are samples uniformly from $\{0,1\}^{|p|}$ and $X_2$ and $Z_2$ are samples uniformly from $\{0,1\}^{n-|p|}$.

Now, consider the following states
\begin{align*}
\ket{\nu_{1}}&= \frac{1}{2^{n/2}}\sum_x \ket{x}\Big(\eta_{g,x}(-1)^{h(x)}\ket{h(x)}\ket{\psi_{x,g}}+\eta_{x,b}(-1)^{\ovl{h(x)}}\ket{\ovl{h(x)}}\ket{\psi_{x,b}}\Big)    
\end{align*}
and
\begin{align*}
    \ket{\nu^p_{2}}= \frac{1}{2^{n/2}}\sum_x \ket{x}(-1)^{x_1\cdot p}&\big(\eta_{g,x}\ket{h(x)}\ket{\psi_{x,g}}+\eta_{x,b}\ket{\ovl{h(x)}}\ket{\psi_{x,b}}\big)
\end{align*}

where $x_1$ is the first $|p|$ bits of $x$.
Let $\ippwc = \big|\braket{\nu_1}{\nu_2^p}\big|^2$.
Naturally, if $\eta_{x,b}=0$, $\ippwc$ directly yields us $\pwc_h(p)$. i.e,
\begin{align*}
    \ippwc &= \Big(\frac{1}{2^n}\sum_{x}(-1)^{h(x)\xor x_1.p}\Big)^2\\
    &= \frac{1}{2^{2n}}\sum_{x,z\in \ftwon{n}} (-1)^{h(x)\xor h(z)\xor p\cdot x_1\xor p\cdot z_1} = \pwc_h(p).
\end{align*}

However, if $\eta_{x_b}\neq 0$ for some $x$, then the cross terms would push the inner product away from $\pwc_h(p)$.
Here, we show that if one is interested only in an $\epsilon$-estimate of $\pwc_h(p)$, then under certain conditions on $\eta_{x,b}$, an $\epsilon$-estimate of the inner product is not too far away from $\pwc_h(p)$.
More concretely, we show that for an $\epsilon/2$-estimate $\widehat{\ippwc}$ of $\ippwc$,
$$\Big|\widehat{\ippwc} - \pwc_h(p)^2\Big| \le \epsilon$$
with probability at least $1-\delta$ if $\gamma = \max_x\{\eta_{x,b}\} \le \epsilon/8$.

Let $\eta_{x,b}\neq 0$ for some $x$'s. Then, we have
\begin{equation*}
   \begin{split}
       \ippwc = {\frac{1}{2^{2n}}\sum_{x,z\in \ftwon{n}}} (-1)^{p\cdot x_1\xor p\cdot z_1}
       \Big[&\eta^2_{x,g}\eta^2_{z,g}(-1)^{h(x)\xor h(z)}+ \eta^2_{x,b}\eta^2_{z,g}(-1)^{\ovl{h(x)}\xor h(z)}
       \\&+
   \eta^2_{x,g}\eta^2_{z,b}(-1)^{h(x)\xor \ovl{h(z)}}+ \eta^2_{x,b}\eta^2_{z,b}(-1)^{\ovl{h(x)}\xor \ovl{h(z)}}\Big]
   \end{split}
\end{equation*}

This implies
\begin{align*}
   &\ippwc-\pwc_h^2(p)\\
   &= \frac{1}{2^{2n}}\sum_{x,z\in \ftwon{n}} (-1)^{p\cdot x_1\xor p\cdot z_1}\bigg[\eta^2_{x,g}\eta^2_{z,g}(-1)^{h(x)\xor h(z)}+ \eta^2_{x,b}\eta^2_{z,g}(-1)^{\ovl{h(x)}\xor h(z)}+
   \eta^2_{x,g}\eta^2_{z,b}(-1)^{h(x)\xor \ovl{h(z)}} \\
   &+ \eta^2_{x,b}\eta^2_{z,b}(-1)^{\ovl{h(x)}\xor \ovl{h(z)}}\bigg]-\frac{1}{2^{2n}}\sum_{x,z\in \ftwon{n}} (-1)^{h(x)\xor h(z)\xor p\cdot x_1\xor p\cdot z_1}\\
   &= \frac{1}{2^{2n}}\sum_{x,z\in \ftwon{n}} (-1)^{p\cdot x_1\xor p\cdot z_1}\bigg[\eta^2_{x,g}\eta^2_{z,g}(-1)^{h(x)\xor h(z)}+ \eta^2_{x,b}\eta^2_{z,g}(-1)^{\ovl{h(x)}\xor h(z)}+\eta^2_{x,g}\eta^2_{z,b}(-1)^{h(x)\xor \ovl{h(z)}}\\
   &+ \eta^2_{x,b}\eta^2_{z,b}(-1)^{\ovl{h(x)}\xor \ovl{h(z)}}-(-1)^{h(x)\xor h(z)}\bigg]\\
\end{align*}

For any fixed $x,z \in \ftwon{n}$, let 
\begin{align*}
    \Delta_{x,z} &= \eta^2_{x,g}\eta^2_{z,g}(-1)^{h(x)\xor h(z)}+ \eta^2_{x,b}\eta^2_{z,g}(-1)^{\ovl{h(x)}\xor h(z)} +\eta^2_{x,g}\eta^2_{z,b}(-1)^{h(x)\xor \ovl{h(z)}}+ \eta^2_{x,b}\eta^2_{z,b}(-1)^{\ovl{h(x)}\xor \ovl{h(z)}}- (-1)^{h(x)\xor h(z)}.
\end{align*}

Using the equality $1-\eta^2_{x,g}\eta^2_{z,g} = \eta^2_{x,g}\eta^2_{z,b} + \eta^2_{x,b}\eta^2_{z,g} + \eta^2_{x,b}\eta^2_{z,b}$ in the above equation, we get

\begin{align*}
    \Delta_{x,z} &= \eta^2_{x,b}\eta^2_{z,g}\big[(-1)^{\ovl{h(x)}\xor h(z)}- (-1)^{h(x)\xor h(z)}\big]\\
    &+\eta^2_{x,g}\eta^2_{z,b}\big[(-1)^{h(x)\xor \ovl{h(z)}}- (-1)^{h(x)\xor h(z)}\big] \\
    &+\eta^2_{x,b}\eta^2_{z,b}\big[(-1)^{\ovl{h(x)}\xor \ovl{h(z)}} - (-1)^{h(x)\xor h(z)}\big]\\
\end{align*}
giving the equation
$$\ippwc - \pwc^2_h(p) = \frac{1}{2^{2n}}\sum_{x,z\in \ftwon{n}}\Delta_{x,z}.$$

Notice that for any $a,b,c,d\in \{0,1\}$, $-2\le \big[(-1)^{a\xor b}-(-1)^{c\xor d}\big] \le 2$.
From this observation, we get that
$$\Delta_{x,z}\ge -2 \Big(\eta^2_{x,b}\eta^2_{z,g} + \eta^2_{x,g}\eta^2_{z,b} + \eta^2_{x,b}\eta^2_{z,b}\Big)
$$
and
$$
\Delta_{x,z}\le 2 \Big(\eta^2_{x,b}\eta^2_{z,g} + \eta^2_{x,g}\eta^2_{z,b} + \eta^2_{x,b}\eta^2_{z,b}\Big).$$

Now,
\begin{align*}
    &\eta^2_{x,b}\eta^2_{z,g} + \eta^2_{x,g}\eta^2_{z,b} + \eta^2_{x,b}\eta^2_{z,b}\\
    &= \eta^2_{x,b}\eta^2_{z,g} + \eta^2_{z,b}(\eta^2_{x,g} + \eta^2_{x,b})\\
    &= \eta^2_{x,b}\eta^2_{z,g} + \eta^2_{z,b}\\
    &\le \eta^2_{x,b} + \eta^2_{z,b}\\
    &\le 2\gamma.
\end{align*}

The second-last inequality follows since $\eta^2_{z,g}\le 1$ and the last equality follows because $\eta^2_{x,b}\le \gamma~\text{and}~\eta^2_{z,b}\le \gamma$. This gives us that $-4\gamma \le \Delta_{x,z} \le 4\gamma$ implying
$$-4\gamma \le \frac{1}{2^{2n}}\sum_{x,z \in \ftwon{n}}\Delta_{x,z} \le 4\gamma.$$
Or, $$\big|\ippwc-\pwc^2_h(p)\big| \le 4\gamma.$$

Now, if $\gamma \le \epsilon/8$, then $4\gamma \le \epsilon/2$. Then, for any $\epsilon/2$-estimate of $\ippwc$, we have,
$$\big|\widehat{\ippwc}-\pwc^2_h(p)\big| \le \big|\widehat{\ippwc}-\ippwc\big| + \big|\ippwc-\pwc^2_h(p)\big| \le \epsilon.$$

Now, we show that if $\delta'<\delta\tau^2/4n$, then the probability that this algorithm fails is at most $\delta$.
The error induced due to estimation is at most $\delta'$.
The number of candidate prefixes at any level for which estimates are obtained is at most $2/\tau^2$.
Using union bound on errors, the error at any level is at most the sum of errors due to the estimation and the amplification routines.
This gives us $\delta_{level} \le \frac{2\delta'}{\tau^2} + \frac{\delta}{2n}$.
Hence, the total error of the algorithm at most $n\cdot \big(\frac{2\delta'}{\tau^2} + \frac{\delta}{n2}\big) = \frac{2n\delta'}{\tau^2} + \frac{\delta}{2}$.
Setting $\delta'\le \frac{\delta\tau^2}{4n}$, the upper bound on the total error is $\delta$.


\section{Quantum Decision Tree Learning: Agnostic Setting}
\label{sec:QDTLagnapp}
We detail the steps of \cref{alg:QWeakLearner} as follows:
\begin{enumerate}
    \item We start with the state $\sum_{x,y}\sqrt{\mathcal{D}_{x,y}}\ket{x}\ket{y}$. Assuming a uniform marginal distribution over $\mathcal{X}$, this can be written as $\frac{1}{\sqrt{2^n}}\sum_{x}\ket{x}\left(\sum_y\alpha_{y|x}\ket{y}\right)$.
    \item We make $k=\bigO{\log{\frac{1}{\gamma}}}$ independent estimations using \cref{lem:MultiEstimation} (M.A.E.) with parameters $(\varepsilon,1-8/\pi^2)$ to obtain the state \\
    \begin{equation*}
        \frac{1}{\sqrt{2^n}}\sum_x\ket{x}{\left(\sum_y\alpha_{y|x}\ket{y}\right)\left(\beta_{gx}\ket{\Tilde{\alpha}_{1|x}}+\beta_{bx}\ket{\mathrm{Err}}\right)^{\otimes k}}.
    \end{equation*}
    We note here that we want to set the value of $h(x)$ as the label in the third register with the larger conditional probability. 
   
    \item On each of the $k=\bigO{\log{\frac{1}{\gamma}}}$ registers, perform thresholding to obtain
    \begin{equation*}
        \frac{1}{\sqrt{2^n}}\sum_x\ket{x}\left(\Hat{\beta}_{gx}\ket{h(x)}\ket{\psi^{\prime}(x)}+\Hat{\beta}_{bx}\ket{\overline{h(x)}}\ket{\psi^{\prime\prime}(x)}\right)^{\otimes k}.
    \end{equation*}
    \item Perform majority on $\bigO{\log{\frac{1}{\gamma}}}$ copies of $\ket{h(x)}$.
    \item Let the product of unitaries from steps 1 to 5 be denoted as $\abpo$. Run \cref{algo:iglalgo} with the oracle $\abpo$ and accuracy and error parameters as $\epsilon$ and $\delta$ to obtain a string $\tilde{S}$.
    \item Return $\chi_{\tilde{S}(x)}$ as our desired weak learner.
\end{enumerate}


\subsection{Proofs of Agnostic Setting}
\label{sec:QDTLagnappproofs}
We now state the following claims, which prove the correctness and give us the query and time complexity of \cref{alg:QWeakLearner}. First, we restate \cref{clm:qgl}, which is proven in \cref{sec:qgl}.

\clmqgl*
\begin{restatable}{clm}{clmiglquerycomplexity}
\cref{alg:QWeakLearner} performs $\tildeO{\frac{n}{\eta\kappa^3}\cdot\log{\frac{1}{\kappa}}}$ queries to $\qaex$ using the $\qgl$ algorithm (\cref{alg:qgl}) where $\kappa$ is the accuracy parameter. The time complexity for \cref{alg:QWeakLearner} is the same as its query complexity with a logarithmic overhead.\label{clm:clmiglquerycomplexity} 
\end{restatable}

\begin{restatable}{clm}{clmmodeestimate}
$\iglalgo$ (\cref{algo:iglalgo}) produces $\Tilde{S}$ such that $\abs{\cor{\chi_{\Tilde{S}}(x)}{D}-\mathrm{max}_S\;\cor{\chi_{S}(x)}{D}}\leq \kappa$. \label{clm:clmmodeestimate}   
\end{restatable}
The proofs for \cref{clm:clmiglquerycomplexity} and \cref{clm:clmmodeestimate} follows directly from \cref{clm:qgl} and \cref{algo:iglalgo}.

\begin{lem}[\citet{kushilevitz1991learning}]
Given a size-$t$ decision tree $f$, the $L_1$ norm of its support is upper-bounded by $t$, i.e., $\sum_{S}\abs{\hat{f}(S)}\leq t$. Such a function $f$ is said to be $t$-sparse.\label{lem:km}
\end{lem}

\begin{restatable}{clm}{clmweaklearner}
The parity monomial $\chi_{\Tilde{S}}$ produced by \cref{alg:QWeakLearner} is a weak agnostic learner.\label{clm:weaklearner}
\end{restatable}
\begin{proof}
    Let $\mathcal{C}$ be a family of size-$t$ decision trees, and let $c\in\mathcal{C}$ be the optimal classifier. Using the Fourier expansion of $c$ and applying \cref{def:correlation} we have 
    $$\cor{c(x)}{D}=\sum_{S\subseteq[n]}\hat{c}(S)\cor{\chi_S(x)}{D}.$$ 
    From \cref{lem:km} we have $\sum_{S\subseteq[n]}\abs{\hat{c}(S)}\leq t$. 
    Using an averaging argument, we have
    \begin{equation}
    \label{eq:maxcor}
        \mathrm{max}_S\abs{\cor{\chi_S(x)}{D}}\geq\frac{1}{t}\cor{c(x)}{D}.
    \end{equation}
    Given any estimated mode $\Tilde{S}$ such that
    \begin{equation*}
        \abs{\cor{\chi_{\Tilde{S}}(x)}{D}-\mathrm{max}_S\;\cor{\chi_{S}(x)}{D}}\leq \kappa
    \end{equation*}
    using \cref{eq:maxcor}, we have 
    \begin{equation*}
        \cor{\chi_{\Tilde{S}}(x)}{D}\geq \frac{1}{t}\cor{c(x)}{D}-\kappa.
    \end{equation*}
    From \cref{def:corrAgnostic}, we see that this is indeed an $\left(m,\kappa,\frac{1}{t}\right)$-weak quantum agnostic learner w.r.t. $c$.
\end{proof}
\cref{clm:clmiglquerycomplexity} gives us the final query complexity and runtime for \cref{alg:QWeakLearner} as stated in \cref{thm:weakagnosticdtlearner}. \cref{clm:clmmodeestimate}, and \cref{clm:weaklearner} guarantee that \cref{alg:QWeakLearner} produces a weak learner for size-$t$ decision trees in polynomial running time. We restate \cref{thm:weakagnosticdtlearner} below for completeness.
\thmwklearner*

\subsection{Proofs of Realizable Setting}
\label{sec:QDTLrealapp}
It is well known that the output state of the Fourier Sampling algorithm can be given as $\ket{\psi} = \sum_S \hat{f}(S)\ket{S}$.  
Measuring the state $\ket{\psi}$ yields subset $S$ with probability ${\hat{f}(S)}^2$. We use $1/\varepsilon^2$ queries to the Fourier sampling oracle $O_f$ to estimate the mode $\Tilde{S}$ of the output distribution with $\varepsilon$ error. This yields the $\chi_{\Tilde{S}}$ term of \cref{clm:weaklearner}.

\begin{restatable}{clm}{clmbayes}
Any weak agnostic learner w.r.t. $h$ obtained by \cref{alg:QWeakLearner} is also a weak agnostic learner w.r.t to the Bayes optimal predictor $f_{\mathcal{B}}$.\label{clm:bayesagnostic}  
\end{restatable}
\begin{proof}
    Using \cref{lem:MultiEstimation}, we have that $\abs{{\alpha}_{1|x}-\Tilde{\alpha}_{1|x}}\leq \varepsilon$, for some $\varepsilon>0$. In \cref{alg:QWeakLearner}, { we set $h(x)=\mathbb{I}\left[\Tilde{\alpha}_{1|x}>1/\sqrt{2}\right]$.}
    {Therefore, we have $\abs{\err{h}{\mathcal{D}}-\err{f_{\mathcal{B}}}{\mathcal{D}}}\leq 2\varepsilon$.} 
    This implies that $\abs{\cor{f}{\mathcal{D}}-\cor{f_{\mathcal{B}}}{\mathcal{D}}}\leq 4\varepsilon$ or $\cor{f}{\mathcal{D}}\in \left[\cor{f_{\mathcal{B}}}{\mathcal{D}}-4\varepsilon,1\right]$. The upper bound is $1$ since the Bayes predictor is the optimal predictor. Therefore given $h$ s.t., $\cor{h}{\mathcal{D}}\geq\eta\cdot\cor{f}{\mathcal{D}}-\kappa^{\prime}$, we have $\cor{h}{\mathcal{D}}\geq\eta\cdot\cor{f_{\mathcal{B}}}{\mathcal{D}}-\kappa$ for appropriate $\kappa^{\prime},\kappa>0$.
\end{proof}
The Bayes predictor $f_{\mathcal{B}}$ is the optimal predictor on a joint distribution $\mathcal{D}$ over $\mathcal{X}\times\{0,1\}$, and defined as $f_{\mathcal{B}}(x)={\mathrm{argmax}_{y\in\{0,1\}}}\,\underset{\mathcal{D}}{\mathrm{Pr}}\left[y|x\right]$, $\forall x\in\mathcal{X}$.
\begin{clm}
$\chi_{\Tilde{S}}$ is a weak realizable learner for size-$t$ decision trees.
\end{clm}
\begin{proof}
    From \cref{clm:bayesagnostic}, we know that {$\cor{\chi_{\Tilde{S}}}{D}\geq\frac{1}{t}\cor{f}{D}-\kappa$}. 
    For the realizable setting, $\cor{f}{D}=1$. 
    Therefore by setting , $\err{\chi_{\Tilde{S}}}{D}\leq\frac{1}{2}-\bigO{\frac{1}{n}}$ we prove that $\chi_{\Tilde{S}}$ is a weak realizable learner for size-$t$ decision trees.
\end{proof}

\section{Discussion on ~\citet{biasedoraclerudyraymond}}
\label{sec:Iwama discussion}
~\citet{biasedoraclerudyraymond} showed that for any $T$ query quantum algorithm $A$ that solves a problem with error at most $\delta$ using a \textit{perfect oracle}, there exists an $O\left(T/\varepsilon\right)$ query algorithm $A^{\prime}$ that solves the same problem with error at most $\delta/6$ using an $\varepsilon$-biased oracle. Note that here we are not referring to strongly-biased oracles.
 
 Let us assume that the oracle invoked by A is perfect. Then if a $T$-query algorithm $A$ solves a problem with error at most $\delta<1/2$, then it is possible to construct an algorithm to solve the same problem with error at most $\delta^{\prime}$ by taking the majority of $O\left(\frac{8(1-\delta)}{(1-2\delta)^2}\log{ (1/\delta^{\prime})}\right)$ invocations of $A$.  
 
 In the case of an $\varepsilon$-biased oracle, the oracle outputs the correct value with probability $1/2+\varepsilon$. If one tries to directly use $A$, since errors add up linearly in quantum~\citep{bernstein1993quantum}, the errors at each step of $A$ will add up to $O(T\varepsilon)$.
 
 Alternatively, one can perform some $k$ many invocations of the biased oracle, obtain the majority, and use the value of the majority as the oracle output. This will serve as an “almost” perfect oracle. If the error at each step is bounded to at most $\delta/T$, then we obtain an algorithm that solves the problem with error at most $\delta$. 
 { If we were to bound the error due to the oracle at each step to at most $\delta$, then we need to find the right value of $k$. Since the oracle outputs the correct value with probability $1/2 + \varepsilon$, using Hoeffding’s inequality, we can obtain the right value of $k$ as $k = \Omega(\log (1/\delta)/\varepsilon^2 )$.} This would increase the query complexity of the algorithm to $\Tilde{O}(T/\varepsilon^2)$. On the other hand, ~\citet{biasedoraclerudyraymond} showed that the same problem can be solved using just $\Tilde{O}(T/\varepsilon)$ queries.
\clearpage

\end{document}